\documentclass[11pt]{article}
\usepackage{amsmath,amssymb}
\usepackage{authblk}
\usepackage{color}
\usepackage{mathrsfs}
\usepackage{mathptmx}
\usepackage{verbatim}
\usepackage{epsfig}
\usepackage{CJK}
\usepackage{bm}
\usepackage{url}
\usepackage{amsfonts}
\usepackage{amsthm}
\input{epsf.sty}
\usepackage{epsfig}
\usepackage[scr=boondox]{mathalfa}
\usepackage{framed}
\usepackage{ulem}
\usepackage{appendix}

\allowdisplaybreaks

\def\<{\langle}
\def\>{\rangle}

\def\be{\begin{equation}}
\def\ee{\end{equation}}
\def\ba{\begin{array}}
\def\ea{\end{array}}

\newtheorem{theorem}{Theorem}[section]
\newtheorem{lemma}{Lemma}[section]
\newtheorem{proposition}{Proposition}[section]

\theoremstyle{definition}
\newtheorem{remark}{Remark}[section]
\newtheorem{example}{Example}[section]

\numberwithin{equation}{section}

\usepackage{amsmath}
\usepackage{amsthm}
\usepackage{amsfonts}
\usepackage{pifont}
\usepackage{color}
\usepackage{bbm}
\usepackage{CJK}
\usepackage{mathrsfs}
\usepackage{fancyhdr}
\usepackage{graphicx}
\usepackage{pdfsync}
\usepackage{psfrag}
\usepackage{time}
\usepackage{txfonts}
\usepackage{tikz-cd}

\usepackage{stmaryrd}
\usepackage{mathrsfs}
\usepackage{txfonts}
\usepackage{amsfonts}

\def\be{\begin{equation}}
\def\ee{\end{equation}}
\def\br{\begin{eqnarray}}
\def\er{\end{eqnarray}}

\title{Variational Principles for Shock Dynamics in Compressible Euler Flows}

\author{Fran\c{c}ois Gay-Balmaz\thanks{Division of Mathematical Sciences, Nanyang Technological University, 21 Nanyang Link, Singapore 637371; 
 e-mail: \tt{francois.gb@ntu.edu.sg} }~
 and Cheng Yang\thanks{Division of Mathematical Sciences, Nanyang Technological University, 21 Nanyang Link, Singapore 637371; 
 e-mail: \tt{cheng.yang@ntu.edu.sg} }
\\
}

\date{\today}

\begin{document}

\maketitle

\begin{abstract}
Hamilton’s principle plays a central role in fluid mechanics as a fundamental tool for deriving governing equations, analyzing conservation laws, and designing structure-preserving numerical schemes. However, its classical formulation is restricted to smooth solutions and does not directly accommodate shock discontinuities. Addressing this limitation within a variational framework has long been a challenging issue. In this paper, we develop a variational framework that extends Hamilton’s principle to shock solutions in compressible fluid dynamics. For the barotropic Euler equations, we introduce a modified action principle that incorporates additional contributions localized at discontinuities. This allows the Rankine–-Hugoniot conditions for mass and momentum to emerge directly from unrestricted variations, without imposing continuity across shocks. The additional term admits a natural interpretation as a dissipation potential, linking the variational structure to energy loss at shocks. 
We then extend the approach to the full compressible Euler equations. Using a variational formulation of nonequilibrium thermodynamics together with suitable variational and phenomenological constraints, we recover the Rankine--Hugoniot relations for mass, momentum, and energy. This yields a unified variational description of shock dynamics in compressible fluids and highlights a fundamental distinction between barotropic and full compressible models in the treatment of energy and entropy at discontinuities.
\end{abstract}

\tableofcontents

\section{Introduction}

Shock waves are singular solutions of compressible fluid models characterized by discontinuities across codimension-one moving interfaces. For the Euler equations, such discontinuities are constrained by the Rankine–Hugoniot conditions, which encode the balance of mass, momentum, and energy across the shock surface. These relations are fundamental in the theory of hyperbolic conservation laws \cite{cofr,lali,whit} and provide the admissibility conditions for weak solutions \cite{evan}.

The Rankine--Hugoniot jump conditions originate from the classical works of Rankine and Hugoniot \cite{Rankine1870,Hugoniot1887,Hugoniot1889}, and arise in the modern theory as compatibility conditions for weak solutions of conservation laws; see also \cite{Lax1957,Godunov1959}. The classical derivation of the Rankine–Hugoniot conditions relies on integral balance laws or distributional formulations of the governing equations. While mathematically robust, this approach does not arise from a variational principle and therefore does not directly reflect the geometric and energetic structure of continuum mechanics. In particular, the incorporation of moving discontinuities into Lagrangian formulations remains subtle.

The variational derivation of the equations of motion for compressible inviscid flow began with Herivel \cite{Herivel1955}, whose early application of Hamilton’s principle was later refined by Serrin \cite{Serrin1959} to account for non-isentropic regimes. Serrin observed that the utility of the variational approach lies in its ability to manage kinematic constraints that exceed the reach of the standard PDEs, effectively bridging the gap between smooth flows and the integral conservation laws required at shock fronts. While these jump conditions are algebraically defined by the Rankine--Hugoniot equations \cite{Taub1948}, their variational representation requires a careful treatment of the underlying configuration space. The work of Taub \cite{Taub1954} and Newcomb \cite{Newcomb1962} extended variational methods to relativistic fluids and magnetohydrodynamics, both within the setting of smooth dynamics. Similarly, the Lie group theoretic geometric approach initiated by Arnold \cite{Arnold1965}, based on the diffeomorphism group formulation, is developed from the outset under the assumption of smooth fluid motion.

Recent work has addressed related aspects of nonsmooth and interface dynamics. Multisymplectic formulations for systems with shock discontinuities have been proposed in \cite{fema}, while constrained variational principles for smooth continua with free boundaries have been developed via Lagrangian reduction in \cite{gama}. In \cite{IzKh} a geodesic, group-theoretic, and
Hamiltonian framework is developed for inviscid
incompressible fluid
flows with vortex sheets. However, a direct variational formulation of shock dynamics for compressible Euler systems, long sought in the literature, remains unavailable.

The purpose of this paper is to formulate variational principles for shock solutions of the compressible Euler equations and to derive the associated Rankine–Hugoniot conditions directly from stationarity of the action. Our approach treats the shock surface, denoted $\Gamma(t)$, as an evolving interface separating two smooth flow regions $M_\pm(t)$ and allows for independent variations on either side of the discontinuity.
We consider both the compressible barotropic Euler equations and the full compressible Euler system with entropy.

For the barotropic case, we introduce the critical action principle:
\[
\begin{aligned}
&\delta \int_{t_0}^{t_1} \left[\int_{M_-(t)}\left(\mathscr l (u_-,\rho_-) +\lambda_- + \rho_- D_t w \right){\rm d}x +\right.\left. \int_{M_+(t)}\left(\mathscr l (u_+,\rho_+) +\lambda_+ + \rho_+ D_t w \right){\rm d}x \right]{\rm d}t=0
\end{aligned}
\]
with respect to arbitrary variations $\delta\Gamma$, $\delta\rho_\pm$, $\delta w$ and constrained variations of the velocity fields 
$$
\delta u_\pm = \partial_t\xi_\pm + \mathcal{L}_{u_\pm}\xi_\pm,\;\;\; \delta\lambda_\pm = - \operatorname{div}(\lambda_\pm \xi_\pm),
$$ 
and continuity conditions for the derivatives $\partial_t w$ and $\nabla_x w$ of the auxiliary field $w$ across the shock interface $\Gamma(t)$.

Allowing variations of the bulk fields (the term bulk fields refers to fields defined on the smooth subdomains $M_\pm(t)$, in contrast to interface quantities defined on the shock surface) and of the interface geometry yields the Rankine–Hugoniot conditions for mass and momentum. The induced energy balance takes the form
$$
\frac{d}{dt}\mathscr E=\int_{\Gamma} (v_s\left[\!\left[\lambda\right]\!\right]-\left[\!\left[\lambda u\right]\!\right]\cdot n )\,{\rm d}a=-\frac{d}{dt}\mathcal{V}_{\rm shock},
$$
where $\mathcal V_{\rm shock}$ is an interface contribution depending on the Lagrangian flow maps via its domain of integration
$$
-\mathcal{V}_{\rm shock}(\varphi_\pm, M_\pm) = \int_{\varphi_-(t)^{-1}(M_-(t))}\Lambda_-(X)\;{\rm d}X + \int_{\varphi_+(t)^{-1}(M_+(t))}\Lambda_+(X)\;{\rm d}X.
$$
In general, this potential is not proportional to the volume functional (i.e., we cannot choose $\Lambda(X)=1$), as can be verified explicitly for one-dimensional barotropic flows. Consequently, the variational formulation leads to a modified energy balance when full interface variations are permitted.

For the full compressible Euler system, we introduce the critical action principle:
\[
\begin{aligned}
&\delta \int_{t_0}^{t_1}
\left[
    \int_{M_-(t)}
        \left( \mathscr l (u_-,\rho_-,s_-)
        + \rho_- D_t w_-
        + (s_- - \varsigma_-) D_t \gamma_- \right)\, {\rm d}x
\right. \\
&\qquad \left.
   + \int_{M_+(t)}
        \left( \mathscr l (u_+,\rho_+,s_+)
        + \rho_+ D_t w_+
        + (s_+ - \varsigma_+) D_t \gamma_+ \right)\, {\rm d}x
\right] {\rm d}t = 0,
\end{aligned}
\]
subject to the advection constraints
\[
\bar D_t\varsigma_\pm=0,\qquad \bar D_\delta\varsigma_\pm=0.
\]
These constraints can also incorporate an entropy flux $j_s$, which may include contributions from heat conduction and other irreversible mechanisms.

In this setting, stationarity with respect to general bulk and interface variations yields the Rankine–Hugoniot conditions for mass, momentum, energy, and the entropy-related variable $s-\varsigma$. Moreover, the total energy satisfies the exact conservation law
$$
\frac{d}{dt}\mathscr E=0.
$$

This reveals a fundamental structural distinction between the barotropic and full compressible Euler systems at the variational level. In the barotropic case, the absence of entropy degrees of freedom necessitates the introduction of an additional interface contribution in order to recover the full Rankine–Hugoniot relations under general variations, resulting in a modified energy balance. Specifically, a dissipative potential $\mathcal{V}_{\rm shock}$ is incorporated into the Lagrangian within Hamilton's principle. By contrast, in the full compressible system, the entropy-related variable provides sufficient flexibility to accommodate general interface variations while preserving energy conservation. In this setting, Hamilton's principle can be consistently extended using a variational formulation for nonequilibrium thermodynamics, as developed in \cite{GBYo2019}.

The framework developed here embeds shock dynamics into classical Lagrangian mechanics and provides a geometric interpretation of the Rankine–Hugoniot conditions. It establishes a foundation for further analytical developments, including structure-preserving discretizations, variational approximations of discontinuous solutions, and connections with thermodynamic admissibility criteria.

\paragraph{Content of the paper.}
In \S\ref{sec: pre}, we review the Euler fluid systems in conservative form and recall the associated Rankine--Hugoniot conditions, emphasizing the key differences between compressible and barotropic flows.  
In \S\ref{sec_3}, we develop the underlying geometric framework in terms of an infinite-dimensional configuration space and formulate a variational principle for the barotropic Euler equations in the presence of shocks. From this principle, the bulk equations, the Rankine--Hugoniot conditions across discontinuities, and the boundary conditions all arise as critical point conditions. Both fluids in fixed domains and fluids with free boundaries are considered.  
In \S\ref{section_EB}, we provide an interpretation of this variational principle in terms of a dissipation potential, thereby clarifying the energy balance for barotropic fluids with shocks. The approach is illustrated in the one-dimensional setting, which also sheds light on the role of the choice of the potential.  
In \S\ref{sec_5}, we extend the analysis to the full compressible Euler system. This setting naturally falls within the framework of nonequilibrium thermodynamics and requires a corresponding extension of the variational formulation to account for irreversible processes while preserving total energy. The bulk equations, Rankine--Hugoniot conditions for mass, momentum, and energy, and the boundary conditions are obtained as critical point conditions.   Finally, in \S\ref{sec:entropy}, we demonstrate that the proposed variational framework can accommodate irreversible processes in fluids by considering heat conduction. An appendix provides background material on Rankine--Hugoniot conditions and weak solutions of conservation laws.

\section{Preliminaries: shocks in compressible and barotropic Euler equations}\label{sec: pre}

In this section, we present the Euler fluid systems in conservative form together with their associated Rankine--Hugoniot conditions, and discuss the key distinctions between compressible and barotropic models. For standard introductions to the compressible and barotropic Euler equations, see \cite{Dafermos,lali,Lax,Majda,Serre}.

\subsection{The compressible Euler equations}

We consider the compressible Euler equations describing the motion of an inviscid fluid. Let
\(\rho(x,t)\) denote the mass density,
\(u(x,t)\in\mathbb R^n\) the velocity field,
and $S(x,t)$ the specific entropy. The governing equations take the conservative form
\begin{align}
\partial_t \rho + \nabla\cdot(\rho u) &= 0, 
\label{euler1}\\
\partial_t(\rho u) + \nabla\cdot(\rho u\otimes u) + \nabla p &= 0,
\label{euler2}\\
\partial_t E + \nabla\cdot\big((E+p)u\big) &= 0,
\label{euler3}
\end{align}
where the total energy density is given by
\[
E=\tfrac12 \rho|u|^2 + \rho e(\rho,S),
\]
and \(e(\rho,S)\) denotes the specific internal energy. Here $p= \rho^2\frac{\partial e}{\partial\rho}$ is the pressure, \(\rho u\otimes u\) is the momentum flux tensor and \(\nabla p = \nabla\cdot(p \mathbf I)\), with \(\mathbf I\) the identity matrix.

Writing the compressible Euler equations in conservative form \eqref{euler1}--\eqref{euler3} highlights the fundamental conservation laws of mass, momentum, and energy, which remain valid even for weak solutions with discontinuities such as shocks.

\paragraph{Entropy considerations.} For smooth solutions, the energy equation \eqref{euler3} can be rewritten as an internal energy balance
\begin{equation}\label{euler4}
\partial_t e + u\cdot\nabla e + \frac{p}{\rho}\nabla\cdot u=0,
\end{equation}
or, equivalently,
\begin{equation}\label{euler4alt}
\partial_t(\rho e)+\nabla\cdot(\rho e\,u)=-p\nabla\cdot u.
\end{equation}
Using the material derivative and the thermodynamic identity
\(
de=T\,dS+\frac{p}{\rho^2}d\rho
\)
together with the continuity equation \eqref{euler1}, we obtain
\[
\frac{de}{dt}=T\frac{dS}{dt}-\frac{p}{\rho}\nabla\cdot u.
\]
Substituting this into \eqref{euler4} gives the entropy transport equation
\(
T\frac{dS}{dt}=0,
\)
or equivalently,
\begin{equation}\label{euler5}
\partial_t S+u\cdot\nabla S=0,
\end{equation}
or in conservative form,
\begin{equation}\label{euler5_div}
\partial_t(\rho S)+\nabla\cdot(\rho S\,u)=0.
\end{equation}
Therefore, for smooth solutions the specific entropy $S$ is advected along particle trajectories.

When shock discontinuities are present, the entropy density $s=\rho S$ is no longer conserved in the classical sense and entropy is instead produced across shocks in accordance with the second law of thermodynamics. In this case, the conservative Euler system \eqref{euler1}--\eqref{euler3} remains valid in the weak sense, and physically admissible solutions satisfy an entropy inequality.

\paragraph{Rankine--Hugoniot conditions.} We review the derivation of the jump conditions across a shock front.
Let $\Omega\subseteq \mathbb R^{n+1}$ be a spacetime domain, and let $\Sigma\subset\Omega$ be a codimension-one hypersurface representing the shock. We parametrize $\Sigma$ as a moving interface $\Gamma(t)$ with normal
velocity $v_s$ and spatial unit normal $n$.

\begin{figure}[h]
\centering
\begin{tikzpicture}
    % Draw the arbitrary region \Omega
    \draw[fill=lightgray!50] plot[smooth cycle, tension=0.8] coordinates {(0,1) (1,3) (2,4) (4,3.5) (4.5,1.5) (3,0.5) (1.5,0.2)};
    
    % Draw the curve \Sigma that fully separates the region
    \draw[thick] (-0.1,1.5) .. controls (2,3) and (3.5,1) .. (4.8,2.5);
    
    % Draw the regions \Omega_1 and \Omega_2
    \node at (1.5,0.5) {\textbf{\(\Omega_{-}\)}};
    \node at (3,3.5) {\textbf{\(\Omega_{+}\)}};
    
    % Draw label for \Gamma
    \node at (2.5, 2.5) {\(\Sigma\)};
    
    % Draw label for \Omega
    \node at (4.8, 0.5) {\(\Omega\)};
    
\end{tikzpicture}
\caption{ Singular surface (shock front) in a domain}
\label{fig:shock}
\end{figure}
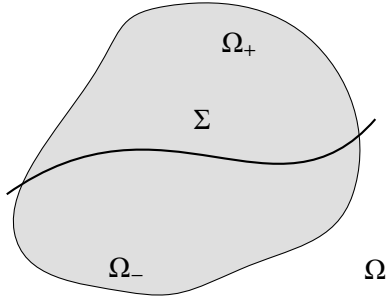

\begin{theorem}[Rankine--Hugoniot conditions for the compressible Euler equations]
\label{thm:RHEuler}
Let $(\rho,u,E)$ be a weak solution of the compressible Euler system
\eqref{euler1}--\eqref{euler3} possessing a jump discontinuity across a smooth,
time-dependent hypersurface $\Gamma(t)\subset \mathbb{R}^n$.
Let $n$ denote the unit normal vector to $\Gamma(t)$, oriented from the
``minus'' side to the ``plus'' side, and let $v_s$ be the normal velocity of
$\Gamma(t)$.

Then the Rankine--Hugoniot jump condition across $\Gamma(t)$ is
\begin{equation}\label{genRH}
n\cdot \big[\!\big[ F(U) \big]\!\big]- v_s\big[\!\big[U \big]\!\big]=0,
\end{equation}
where $U=(\rho,\rho u,E)$ is the vector of conserved
variables, $F(U)$ is the associated flux tensor, and
$[\![\,\cdot\,]\!]$ denotes the jump across $\Gamma(t)$.

Equivalently, the jump conditions \eqref{genRH} can be written componentwise as:
\begin{enumerate}
\item \emph{Mass conservation:}
\begin{equation}\label{eulerRH1}
n\cdot\big[\!\big[ \rho(u-v_s n) \big]\!\big]=0.
\end{equation}

\item \emph{Momentum conservation:}
\begin{equation}\label{eulerRH2}
n\cdot\big[\!\big[ \rho u\otimes(u-v_s n) + p \mathbf I \big]\!\big]=0.
\end{equation}

\item \emph{Energy conservation:}
\begin{equation}\label{eulerRH3}
n\cdot\big[\!\big[ E(u-v_s n) + p u \big]\!\big]=0.
\end{equation}
\end{enumerate}
These conditions characterize the conservation of mass, momentum, and energy
across shock discontinuities in compressible inviscid flow.
\end{theorem}

\begin{proof}
The compressible Euler equations \eqref{euler1}, \eqref{euler2}, and \eqref{euler3}
can be written in conservative form as
\begin{equation}\label{diveq}
\partial_t U + \nabla \cdot F(U) = 0,
\end{equation}
where \(U = (\rho,\rho u,E)\) is the vector of conserved variables and
\(F(U)\) is the corresponding flux.

The associated integral conservation law reads
\begin{equation}\label{diveq2}
\frac{d}{dt}\int_D U{\rm d}x
=
- \int_{\partial D} F(U) \cdot n \,{\rm d}a,
\end{equation}
for any fixed control volume \(D \subset \mathbb{R}^n\), possibly intersecting the
moving singular surface \(\Gamma(t)\).

Applying Lemma \ref{lem:cons_law_RH} to \eqref{diveq2} yields the jump relation
\begin{equation}\label{genRH2}
[\![ F(U) ]\!] \cdot n
=
v_s \,[\![ U ]\!],
\end{equation}
which is equivalent to \eqref{genRH}. Substituting the explicit expressions for
\(U\) and \(F(U)\) corresponding to the compressible Euler
system then yields the componentwise jump conditions
\eqref{eulerRH1}, \eqref{eulerRH2}, and \eqref{eulerRH3}.
\end{proof}

\subsection{The Barotropic Euler equations}

The compressible barotropic Euler equations describe the motion of a compressible fluid for which the pressure depends only on the density,
\(
p=p(\rho).
\)
In this setting no separate entropy or energy equation is required, and the dynamics are governed by the conservation of mass and momentum.

Letting $\rho(x,t)$ denote the density and
$u(x,t)\in\mathbb{R}^n$ the velocity field, the equations take the conservative form
\begin{equation}\label{eq:baro_nd}
\begin{aligned}
\partial_t \rho + \nabla\cdot(\rho u) &= 0,\\
\partial_t(\rho u) + \nabla\cdot(\rho u\otimes u) + \nabla p(\rho) &= 0.
\end{aligned}
\end{equation}

\begin{remark}[Mechanical energy conservation]\label{rmk:energy_conservation}
For barotropic flows, the relevant conserved quantity is the mechanical energy.  
Define the energy density
\[
E = \tfrac12 \rho |u|^2 + \rho e(\rho),
\qquad
e(\rho) := \int^{\rho} \frac{p(r)}{r^2}\,{\rm d}r .
\]
For smooth solutions of \eqref{eq:baro_nd}, the energy density $E$ satisfies the conservation law
\begin{equation}\label{eq:baro_energy}
\partial_t E
+ \nabla\cdot\big((E+p(\rho))u\big)=0 .
\end{equation}
Although no thermodynamic variables appear explicitly, the function $e(\rho)$ plays a role analogous to internal energy in the isentropic setting.
\end{remark}

By applying Lemma~\ref{lem:cons_law_RH} to the integral conservation law associated with \eqref{eq:baro_nd}, i.e., with $U = (\rho, \rho u)$, yields the
Rankine--Hugoniot jump conditions for the barotropic Euler equations:
\begin{enumerate}
\item \emph{Mass conservation:}
$$
n\cdot\big[\!\big[ \rho(u-v_s n) \big]\!\big]=0.
$$

\item \emph{Momentum conservation:}
$$
n\cdot\big[\!\big[ \rho u\otimes(u-v_s n) + p \mathbf I \big]\!\big]=0.
$$
\end{enumerate}   
%\end{remark}

\section{Variational formulation for the barotropic Euler equations with shocks}\label{sec_3}

This section establishes a variational principle for the barotropic Euler equations in the presence of shocks, from which the bulk equations, the Rankine–Hugoniot conditions across discontinuities, and the boundary conditions arise as critical point conditions. We consider both fluids evolving in a fixed domain and fluids with free boundaries. Our approach is valid for general Lagrangian densities.

It is well known that smooth fluid motion in the Lagrangian description can be derived from Hamilton's principle
\begin{equation}\label{HP}
\delta \int_{t_0}^{t_1}L(\varphi, \dot \varphi ) \,{\rm d}t=0, 
\end{equation}
where the Lagrangian $L:TQ\rightarrow\mathbb{R}$ is given in terms of a Lagrangian density by
\begin{equation}\label{L_density}
L(\varphi, \dot{\varphi})= \int_M\mathcal{L}(\varphi, \dot \varphi, \nabla\varphi, \varrho)\, {\rm d}X .
\end{equation}
Here, $Q$ denotes the (infinite-dimensional) configuration manifold of the fluid. For a fluid moving in a fixed domain $M$, one has $Q=\operatorname{Diff}(M)$, the diffeomorphism group of $M$; for a free-boundary fluid, $Q=\operatorname{Emb}(M, \mathbb{R}^n)$ the manifod of embeddings.

The Lagrangian density in \eqref{L_density} depends on the mass density $\varrho$ in Lagrangian coordinates. In the smooth setting, this density is time-independent due to the mass conservation constraint and is therefore not varied in the variational principle; it can be prescribed a priori.

To extend this framework to flows with shocks, the first step is to identify the appropriate infinite-dimensional geometric setting.

\subsection{Geometric setting}

\paragraph{Fluids in a fixed domain.}
We assume that the fluid occupies a fixed spatial domain $M\subset\mathbb{R}^n$, diffeomorphic to a closed ball, while a time-dependent
shock surface $\Gamma(t)$ partitions $M$ into two subdomains $M_\pm(t)$.  
The boundary $\partial M$ is fixed, but it is decomposed into two time-dependent
subsets $S_\pm(t)\subset\partial M$ according to the position of the shock.
We consider two labels spaces $\mathcal{B}_\pm$ (each diffeomorphic to $\mathbb{R}^n_+$), corresponding to labels of fluid particles that are associated with $M_\pm(t)$. In presence of shocks, it is crucial to observe that the set of labels corresponding to fluid particle actually in $M_\pm(t)$ at time $t$, i.e. the domain $\varphi_{\pm}(t)^{-1} (M_\pm (t)) \subset \mathcal{B}_\pm$ depends on time. Since the fluid has fixed boundary, the preimages
of the boundary components under the flow maps are constrained as follows:
\[
\varphi_\pm(t)^{-1}(S_\pm(t))\subset \partial\mathcal{B}_\pm.
\]
Thus the labels of particles moving on the physical boundary are fixed once and for all.

Based on these considerations, given a domain $M$ with boundary $\partial M$ and label spaces
$\mathcal{B}_\pm$ (each diffeomorphic to $\mathbb{R}^n_+$), the configuration manifold
for a fluid with fixed boundary and a shock surface is
\begin{equation}\label{Qfix}
Q_{\rm fix}
=\Big\{(\varphi_-,\varphi_+,\Gamma)\in
\operatorname{Diff}(\mathbb{R}^n)\times\operatorname{Diff}(\mathbb{R}^n)\times\mathcal{S}
\;\Big|\;
\varphi_\pm^{-1}(S_\pm)\subset \partial\mathcal{B}_\pm
\Big\},
\end{equation}
where $\mathcal{S}$ denotes the manifold of hypersurfaces that split $M$ into two
subdomains, and where $S_\pm$ are determined by $\Gamma$.

Note that rather than describing the motion by time-dependent embeddings
\[
\varphi_\pm(t)\colon B_\pm(t)\subseteq\mathbb{R}^n \longrightarrow M_\pm(t)\subseteq\mathbb{R}^n
\]
defined on variable reference domains, we adopt a global formulation in which the
motions $\varphi_\pm(t)$ are treated as diffeomorphisms of the entire space
$\mathbb{R}^n$.  
In this approach, the restriction to the physically relevant particles is not
encoded in the configuration space itself but is instead enforced at the level of
the Lagrangian, which selects only those portions of the diffeomorphism corresponding
to particles lying in the smooth regions at a given time.

\paragraph{Free-boundary fluids.}
We now allow the spatial domain occupied by the fluid to evolve in time.
Let $M(t)\subseteq\mathbb{R}^n$ denote the fluid domain at time $t$, with a moving
boundary $\partial M(t)$.  
A shock surface $\Gamma(t)$ partitions $M(t)$ into two subdomains $M_\pm(t)$, and the
boundary $\partial M(t)$ is accordingly decomposed into time-dependent components
$S_\pm(t)\subset\partial M(t)$.

As in the fixed-boundary case, the preimages of these boundary components under the
flow maps are constrained:
\[
\varphi_\pm(t)^{-1}(S_\pm(t))\subset \partial\mathcal{B}_\pm.
\]
This describes the time dependent set consisting of labels of particles moving on the free boundary of the fluid.

The configuration manifold for a free-boundary fluid with a shock surface is then
\begin{equation}\label{Qfree}
Q_{\rm free}
=\Big\{(\varphi_-,\varphi_+,\partial M,\Gamma)\in
\operatorname{Diff}(\mathbb{R}^n)\times\operatorname{Diff}(\mathbb{R}^n)\times\mathcal{BS}
\;\Big|\;
\varphi_\pm^{-1}(S_\pm)\subset \partial\mathcal{B}_\pm
\Big\},
\end{equation}
where $\mathcal{BS}$ denotes the bundle whose elements $(\partial M,\Gamma)$ consist
of a boundary $\partial M$ together with a hypersurface $\Gamma$ that splits the
associated domain $M$ into two parts.  
The boundary components $S_\pm$ are determined by $(\partial M,\Gamma)$.

\subsection{Action functional and variations}\label{action_functional}

We present a variational formulation for the barotropic Euler equations that extend the Hamilton principle \eqref{HP} of the smooth case to the treatment of shocks. It produces the Rankine--Hugoniot jump conditions for mass and momentum as critical conditions, while permitting arbitrary variation of the variables in the configuration manifold. We then consider the variational principle induced in Eulerian coordinates, via relabeling symmetries.

\medskip

The following points are used to construct the appropriate extension of the Lagrangian \eqref{L_density} to the treatment of shocks:
\begin{itemize}
\item[(i)] The Lagrangian involves the sum of two contributions: one for each of the subdomains;
\item[(ii)] Each of these contributions must take into account only of the fluid labels $\varphi_\pm^{-1}(M_\pm(t))$ 
occupying the smooth regions $M_\pm(t)$ at time $t$;
\item[(iii)] As opposed to the smooth case, one cannot a priori fix a time-independent mass density in Lagrangian coordinates as in \eqref{L_density}. It has to be considered as a part of the variational principle and mass conservation has to be enforced weekly in each smooth subdomains, via the introduction of the scalar variable $W$. This leads to the enlargement of the configuration manifold $Q_{\rm fix}$ to $\mathcal{Q}_{\rm fix}$ described below, similarly for $Q_{\rm free}$ and $\mathcal{Q}_{\rm free}$;
\item[(iv)] The inclusion of a potential term, specific to the energetic of the  barotropic shocks, whose physical interpretation will be given later.
\end{itemize}

\paragraph{Hamilton's principle for shocks in Lagrangian coordinates.} From the above considerations, the full configuration manifold in the fixed-domain case is
\[
\mathcal{Q}_{\mathrm{fix}} = Q_{\mathrm{fix}} \times C^{\infty}(\mathbb{R}^{n+1}) \times C^{\infty}(\mathbb{R}^{n+1}) \times C^{\infty}(\mathbb{R}^{n+1})\times C^{\infty}(\mathbb{R}^{n+1}),
\]
and similarly, in the free-boundary case, one defines $\mathcal{Q}_{\mathrm{free}}$ based on $Q_{\mathrm{free}}$. A typical configuration $q \in \mathcal{Q}$ is of the form
\begin{equation}\label{q_variables}
q =
\begin{cases}
(\varphi_-, \varphi_+, \Gamma, \rho_-, \rho_+, W_-,W_+),
& \text{if } q \in \mathcal{Q}_{\mathrm{fix}}, \\[4pt]
(\varphi_-, \varphi_+, \partial M, \Gamma, \rho_-, \rho_+, W_-,W_+),
& \text{if } q \in \mathcal{Q}_{\mathrm{free}}.
\end{cases}
\end{equation}
We will assume that the Eulerian function $w$ on $M$ piecewise defined by $W_-\circ \varphi_-^{-1}$ and $W_+\circ \varphi_+^{-1}$ is a $C^1$ function.

Based on points (i)--(iv), the total Lagrangian is a function
\[
L : T\mathcal{Q} \to \mathbb{R},
\]
where $\mathcal{Q}$ stands for either $\mathcal{Q}_{\mathrm{fix}}$ or $\mathcal{Q}_{\mathrm{free}}$. It is defined by
\begin{equation}\label{Lagrangian}
\begin{aligned}
L(q, \dot{q}) 
&= \int_{\varphi_-^{-1}(M_-(t))}
\Big( \mathcal{L}(\varphi_-, \dot{\varphi}_-, \nabla \varphi_-, \varrho_-) + \varrho_- \dot{W}_- \Big)\, \mathrm{d}X
+ \int_{\varphi_-^{-1}(M_-(t))} \Lambda_-(X)\, \mathrm{d}X \\
&\qquad + \int_{\varphi_+^{-1}(M_+(t))}
\Big( \mathcal{L}(\varphi_+, \dot{\varphi}_+, \nabla \varphi_+, \varrho_+) + \varrho_+ \dot{W}_+ \Big)\, \mathrm{d}X
+ \int_{\varphi_+^{-1}(M_+(t))} \Lambda_+(X)\, \mathrm{d}X,
\end{aligned}
\end{equation}
where, $\Lambda_\pm(X)$ are prescribed functions on the reference configuration.

Hamilton’s principle then takes the usual form
\begin{equation}\label{HP_shocks}
\delta \int_{t_0}^{t_1} L(q, \dot{q}) \, {\rm d}t = 0,
\end{equation}
for arbitrary variations of the curve $q(t)$ with fixed endpoint, where $q(t)$ is as in \eqref{q_variables}.

\paragraph{Hamilton's principle for shocks in Eulerian coordinates.} Passing to the Eulerian description, we denote by $\lambda_\pm$ the Eulerian
counterparts of the function $\Lambda_\pm$.
The critical action principle then takes the form
\begin{equation}\label{critical_AP}
\begin{aligned}
&\delta \int_{t_0}^{t_1} \left[\int_{M_-(t)}\left(\mathscr l (u_-,\rho_-) +\lambda_- + \rho_- D_t w_{-} \right){\rm d}x +\right.\left. \int_{M_+(t)}\left(\mathscr l (u_+,\rho_+) +\lambda_+ + \rho_+ D_t w_{+} \right){\rm d}x \right]{\rm d}t=0
\end{aligned}
\end{equation}
Here $u_\pm$ and $\rho_\pm$ denote the Eulerian velocity and density fields in
$M_\pm(t)$, respectively, $\mathscr l $ is the reduced Lagrangian density, and we introduce the notation
\begin{equation}\label{mat_der}
D_t=\partial_t+u\cdot\nabla
\end{equation}
for the material derivative.
The variations are arbitrary with respect to $\delta\Gamma$, $\delta\rho_\pm$,
and $\delta w$, while the variations of the velocity and $\lambda_\pm$ are
of the usual constrained form
\begin{equation}\label{delta_u_lambda}
\delta u_\pm=\partial_t\xi_\pm+\mathscr{L}_{u_\pm}\xi_\pm,
\qquad
\delta\lambda_\pm=-\nabla\cdot (\lambda_\pm\xi_\pm),
\end{equation}
where
\[
\xi_\pm=\delta\varphi_\pm\circ\varphi_\pm^{-1}
\]
are arbitrary Eulerian vector fields.
Variations of the shock hypersurface $\Gamma$ are induced by normal
displacements of the form $\delta\Gamma=\xi_s\cdot n_+$, where $\xi_s$ is an
arbitrary vector field defined on $\Gamma$ and $n_+$ denotes the unit
normal to the shock surface.

\begin{remark}[Form of the action functional]
Note the presence of the additional term $\varrho_\pm \dot W_\pm$ in the Lagrangian. This term plays a central role in extending Hamilton's principle to the variational formulation of nonequilibrium thermodynamics (\cite{GBYo2019}) and open fluid systems (\cite{ElGBWu2025}), where arbitrary variations of the mass density $\varrho$ must be permitted.
By contrast, for smooth and reversible fluid motion in a closed domain, the mass density in the Lagrangian description may be assumed time-independent from the outset, $\varrho(t,X)= \varrho(X)$, and therefore need not enter explicitly into Hamilton's principle.
In previous works, the inclusion of the term $\varrho_\pm \dot W_\pm$ was essential for modelling diffusion processes (\cite{GBYo2019}) and mass flux across permeable boundaries (\cite{ElGBWu2025}). In the present setting, it is crucial for deriving the Rankine–Hugoniot condition associated with mass conservation. The meaning of the additional term involving $\Lambda(X)$ will be explained later in Section \ref{section_EB}.
\end{remark}

\subsection{Derivation of the Rankine--Hugoniot condition for mass and momentum}\label{sec:RH_mass}

In this section, we derive the bulk equations, the Rankine–Hugoniot conditions across discontinuities, and the associated boundary conditions from the variational principle \eqref{critical_AP}, obtained from Hamilton’s principle \eqref{HP_shocks} in Eulerian coordinates. The derivation is most naturally carried out using Lie derivatives; the resulting equations are then reformulated in standard vector notation in a subsequent paragraph.

\begin{theorem}\label{RHvar_modi}
The critical action principle \eqref{critical_AP} yields the following equations:
\begin{itemize}
\item[(1)] \textbf{Bulk equations.} The fluid equations in the bulk are
\begin{equation}\label{genEulerLag}
\left\{
\begin{aligned}
&\partial_t \rho_\pm + \mathscr L_{u_\pm}\rho_\pm = 0,\\
&\partial_t \frac{\partial\mathscr l}{\partial u_\pm}
+ \mathscr L_{u_\pm}\frac{\partial\mathscr l}{\partial u_\pm}
= \rho_\pm \, d\frac{\partial\mathscr l}{\partial \rho_\pm}.
\end{aligned}
\right.
\end{equation}

\item[(2)] \textbf{Boundary conditions.} 
For free-boundary fluids, the boundary condition on $S_\pm$ is
\begin{equation}\label{genBounCon}
\mathscr l(u_\pm,\rho_\pm) - \rho_\pm \frac{\partial\mathscr l}{\partial \rho_\pm} = 0.
\end{equation}
For fluids in a fixed domain, no additional boundary conditions arise. In particular, the condition $u_\pm \cdot n_\pm = 0$ on $S_\pm$, where $n_\pm$ denotes the outward unit normal vector to $S_\pm$, follows directly from the choice of configuration manifold.

\item[(3)] \textbf{Rankine--Hugoniot condition for momentum.}
On the singular surface $\Gamma$, one has
\begin{equation}\label{genLagRH}
v_s\left[\!\left[\frac{\partial\mathscr l}{\partial u}\right]\!\right](\cdot)
=\left(\left[\!\left[\frac{\partial\mathscr l}{\partial u}\otimes u\right]\!\right]
+\left[\!\left[\mathscr l - \rho\frac{\partial\mathscr l}{\partial \rho}\right]\!\right]\delta\right)\cdot n,
\end{equation}
where $v_s$ denotes the normal speed of $\Gamma$, $\delta$ is the $(1,1)$ identity tensor characterized by $\delta^i_j = 1$ if $i=j$ and $\delta^i_j=0$ otherwise, and $n$ is the unit normal vector field along $\Gamma$.

\item[(4)] \textbf{Rankine--Hugoniot condition for mass.}
On the singular surface $\Gamma$, one has
\begin{equation}\label{eq:RH_mass}
v_s\left[\!\left[\rho\right]\!\right]=\left[\!\left[\rho u\right]\!\right]\cdot n.
\end{equation}
\end{itemize}
\end{theorem}

\begin{proof}
Let $q(t)\in \mathcal{Q}_{\rm free}$ (or $\mathcal{Q}_{\rm fix}$)
be a curve in the configuration manifold and let $q_\epsilon$ be an arbitrary variation with fixed endpoints. 
\iffalse
Recall that the curve $\eta(t)$ encodes all the $\varphi_\pm(t)$, $\partial M(t)$, $\Gamma(t)$, as well as $\varrho_\pm(t)$ and $W(t)$.
\fi
The stationarity condition is given by
\begin{equation}\label{stan_eq_modi}
\begin{aligned}
0=&\delta\int_{t_0}^{t_1}{\rm d}t\int_{M_+}\big(\mathscr{l}(u_+,\rho_+)+\lambda_++\rho_+ D_t w\big)\,{\rm d}x
+\int_{M_-}\big(\mathscr{l}(u_-,\rho_-)+\lambda_-+\rho_- D_t w\big)\,{\rm d}x\\
=&\int_{t_0}^{t_1}{\rm d}t\int_{M_+}
\left(\frac{\partial\mathscr l}{\partial u_+}\cdot\delta u_+
+\frac{\partial\mathscr l}{\partial \rho_+}\delta \rho_+
+\delta\lambda_+
+\delta\rho_+D_t w_+
+\rho_+\delta D_t w_+\right)\,{\rm d}x\\
&+\int_{t_0}^{t_1}{\rm d}t\left(\int_{S_+}
(\mathscr{l}(u_+,\rho_+)+\lambda_++\rho_+ D_t w_+)\,\xi_+\cdot n_+\,{\rm d}a+\int_{\Gamma}
(\mathscr{l}(u_+,\rho_+)+\lambda_++\rho_+ D_t w_+)\,\xi_s\cdot n_+\,{\rm d}a\right)\\
&+\int_{t_0}^{t_1}{\rm d}t\int_{M_-}
\left(\frac{\partial\mathscr l}{\partial u_-}\cdot\delta u_-
+\frac{\partial\mathscr l}{\partial \rho_-}\delta \rho_-
+\delta\lambda_-
+\delta\rho_-D_t w_-
+\rho_-\delta D_t w_-\right)\,{\rm d}x\\
&+\int_{t_0}^{t_1}{\rm d}t\left(\int_{S_-}
(\mathscr{l}(u_-,\rho_-)+\lambda_-+\rho_- D_t w_-)\,\xi_-\cdot n_-\,{\rm d}a+\int_{\Gamma}
(\mathscr{l}(u_-,\rho_-)+\lambda_-+\rho_- D_t w_-)\,\xi_s\cdot n_-\,{\rm d}a\right).
\end{aligned}
\end{equation}
We recall that $\delta\Gamma= \xi_s\cdot n_+$ denotes the variation of the shock interface $\Gamma$
Note that if the domain is fixed, then on $S_\pm$, $\xi_\pm\cdot n_\pm=0$.

We now regard $\Omega=\Omega_+\cup\Omega_-$ as a spacetime domain with singular submanifold $\Sigma$ (see Figure~\ref{fig:shock}). A direct computation, using the variations $\delta u_\pm$ and $\delta \lambda_\pm$ given in \eqref{delta_u_lambda}, yields
\begin{align}\label{first_eq_modi}
%\begin{split}
&\iint_{\Omega_\pm}\left(\frac{\partial\mathscr l}{\partial u_\pm}\cdot\delta u_\pm
+\frac{\partial\mathscr l}{\partial \rho_\pm}\cdot\delta \rho_\pm
+\delta\lambda_\pm+\delta\rho_\pm D_t w_\pm +\rho_\pm\delta D_tw_\pm\right)\,{\rm d}x{\rm d}t\notag\\
=&\iint_{\Omega_\pm}\left(\frac{\partial\mathscr l}{\partial u_\pm}\cdot (\partial_t\xi_\pm+[u_\pm,\xi_\pm])
+\left(\frac{\partial\mathscr l}{\partial \rho_\pm}+D_t w_\pm\right)\delta\rho_\pm
-\mathscr L_{\xi_{\pm}}\lambda_\pm\right)\,{\rm d}x{\rm d}t\notag\\
&\quad +\iint_{\Omega_\pm}\rho_\pm\left(D_t\delta w_\pm+(\partial_t \xi_\pm+\mathscr L_{u_\pm}\xi_\pm,\nabla) w_\pm\right)\,{\rm d}x{\rm d}t\notag\\
=&\iint_{\Omega_\pm}\left(\left(-\partial_t \frac{\partial\mathscr l}{\partial u_\pm}
-\mathscr L_{u_\pm}\frac{\partial\mathscr l}{\partial u_\pm}\right)\cdot\xi_\pm
+\left(\frac{\partial\mathscr l}{\partial \rho_\pm}+D_t w_\pm\right)\delta\rho_\pm\right)\,{\rm d}x{\rm d}t\notag\\
&\quad+\iint_{\Omega_\pm}\left(\bar{D}_t(\rho_\pm\delta w_\pm)
-\bar{D}_t\rho_\pm\delta w_\pm
+\bar{D}_t(\rho_\pm\xi_\pm\cdot\nabla w_\pm)
-\xi_\pm\cdot(\partial_t(\rho_\pm\nabla w_\pm)+\mathscr L_{u_\pm} (\rho_\pm\nabla w_\pm))\right)\,{\rm d}x{\rm d}t\notag\\
&\quad+\int_{\partial\Omega_\pm}\left(\frac{\partial\mathscr l}{\partial u_\pm}\cdot \xi_\pm n_t^\pm
+ \frac{\partial\mathscr l}{\partial u_\pm}(\xi_\pm) \;u_\pm \cdot n_x^{\pm}
-\lambda_\pm \xi_\pm \cdot n_x^\pm\right)\,{\rm d}a(x,t)\notag\\
=&\iint_{\Omega_\pm}\left(\left(-\partial_t \frac{\partial\mathscr l}{\partial u_\pm}
-\mathscr L_{u_\pm}\frac{\partial\mathscr l}{\partial u_\pm}
- (\partial_t(\rho_\pm\nabla w_\pm)+\mathscr L_{u_\pm} (\rho_\pm\nabla w_\pm))\right)\cdot\xi_\pm
+\left(\frac{\partial\mathscr l}{\partial \rho_\pm}+D_t w_\pm\right)\delta\rho_\pm
-\bar{D}_t\rho_\pm\delta w_\pm\right)\,{\rm d}x{\rm d}t\notag\\
&\quad+\int_{\partial\Omega_\pm}\left(\frac{\partial\mathscr l}{\partial u_\pm}\cdot \xi_\pm n_t^\pm
+ \frac{\partial\mathscr l}{\partial u_\pm}(\xi_\pm) \;u_\pm \cdot n_x^{\pm}
-\lambda_\pm \xi_\pm \cdot n_x^\pm\right)\,{\rm d}a(x,t)\notag\\
&\quad+\int_{\partial\Omega_\pm} \left(\rho_\pm\delta w_\pm n_t^\pm
+\rho_\pm\delta w_\pm u_\pm \cdot n_x^\pm
+\rho_\pm\xi_\pm\cdot\nabla w_\pm n_t^\pm
+\rho_\pm\xi_\pm\cdot\nabla w_\pm u_\pm \cdot n_x^\pm\right)\,{\rm d}a(x,t),
\end{align}
where $\bar D_t$ denotes the material derivative of a density:
\begin{equation}\label{mat_der_dens}
\bar D_t= \partial_t + \mathscr{L}_u,
\end{equation}
which corresponds to $\bar D_t(\cdot)= \partial_t(\cdot) + \nabla\cdot (\,\cdot\, u)$ in usual vector calculus notation, see below.
In \eqref{first_eq_modi}, we used the following formulas
\begin{equation}\label{useful_formulas}
\begin{aligned}
\delta D_t w
&=\delta\partial_t w+(u,\nabla)\delta w+(\delta u,\nabla)w
=D_t\delta w+(\partial_t \xi+\mathscr L_{u}\xi,\nabla) w\\
\rho D_t \delta w&=\bar D_t(\rho\delta w)-\bar D_t\rho \delta w\\
\rho(\partial_t \xi+\mathscr L_{u}\xi,\nabla) w
&=\bar{D}_t(\rho\xi\cdot\nabla w)-\xi\cdot(\partial_t(\rho\nabla w)+\mathscr L_u (\rho\nabla w)),
\end{aligned}
\end{equation}
In the integration by parts performed above, we recall that 
$\frac{\partial \mathscr l }{\partial u}$ is interpreted as a $1$-form density rather than as an ordinary $1$-form and $\mathscr L_u \frac{\partial \mathscr l }{\partial u}$ denotes the Lie derivative of a 1-form density.

\medskip

\noindent\textit{Bulk equations.} Plugging \eqref{first_eq_modi} into \eqref{stan_eq_modi}, we get in $\Omega_\pm$:
\begin{align*}
\xi_\pm:&\;\;\;
\partial_t \frac{\partial\mathscr l}{\partial u_\pm}
+\mathscr L_{u_\pm}\frac{\partial\mathscr l}{\partial u_\pm}
=\rho_\pm d \frac{\partial\mathscr l}{\partial \rho_\pm},\\
\delta \rho_\pm:&\;\;\;
D_t w_\pm=-\frac{\partial\mathscr l}{\partial \rho_\pm},\\
\delta w_\pm:&\;\;\;
\bar D_t \rho_\pm=0.
\end{align*}
The form of the RHS of the first equation arises from the following equality
\[
\partial_t(\rho\nabla w)+\mathscr L_u (\rho\nabla w)=\bar{D_t}\rho\nabla w+\rho \nabla D_t w
=-\rho d\frac{\partial\mathscr l}{\partial \rho}.
\]
We hence get Equations \eqref{genEulerLag}.
Recall that we regard $\rho$ as a density rather than as a scalar function, and therefore $\mathscr L_u(\rho\,\nabla w)$ is the Lie derivative of a $1$-form density.

\medskip

Let us split the contribution of $\partial\Omega_\pm$ as follows:
\[
\int_{\partial \Omega_\pm}=\int_{t_0}^{t_1}\int_{S_\pm}+\int_\Sigma.
\]

\noindent\textit{Boundary conditions.} We can get the boundary conditions on $S_\pm$ from the following equation (note that $\lambda_\pm$ are cancelled out on $S_\pm$):
\[
\begin{aligned}
0=&\int_{t_0}^{t_1}\int_{S_\pm}
\bigg(\frac{\partial\mathscr l}{\partial u_\pm}\cdot \xi_\pm n_t^\pm
+ \frac{\partial\mathscr l}{\partial u_\pm}(\xi_\pm) \;u_\pm \cdot n_x^{\pm}
+\left(\mathscr l_\pm-\rho_\pm\frac{\partial\mathscr l}{\partial \rho_\pm}\right) \xi_\pm \cdot n_x^\pm\\
&\;\;+\rho_\pm\xi_\pm\cdot\nabla w_\pm n_t^\pm
+\rho_\pm\xi_\pm\cdot\nabla w_\pm u_\pm \cdot n_x^\pm\bigg)\,{\rm d}a(x,t).
\end{aligned}
\]

For fluids in a fixed domain, along $S_\pm$ one has 
\[
n_t^\pm = -\,u_\pm\cdot n_x^\pm = 0,
\qquad 
\xi_\pm\cdot n_x^\pm = 0 .
\]
Under these conditions, the above equation is automatically satisfied.

For free-boundary fluids, along $S_\pm$ one has only
\[
n_t^\pm = -\,u_\pm\!\cdot\! n_x^\pm.
\]
Therefore, we obtain
\[
\rho_\pm \,\xi_\pm\!\cdot\!\nabla w_\pm \, n_t^\pm
+\rho_\pm \,\xi_\pm\!\cdot\!\nabla w_\pm \, u_\pm \!\cdot\! n_x^\pm
=0
\]
and
\[
\frac{\partial\mathscr l}{\partial u_\pm}\!\cdot\! \xi_\pm \, n_t^\pm
+\frac{\partial\mathscr l}{\partial u_\pm}(\xi_\pm)\, u_\pm \!\cdot\! n_x^\pm
=0.
\]
Consequently, the natural boundary condition along $S_\pm$ reduces to
Equation~\eqref{genBounCon}.

\medskip

\noindent\textit{Rankine–Hugoniot condition for mass.} Assuming that $\delta w_\pm$ is continuous across $\Gamma$,
i.e.\ $\delta w_+=\delta w_-=\delta w$, the boundary terms on $\Gamma$ imply
\[
\delta w:\;\;\;
\rho_+ n_t^++\rho_+ u_+\cdot n_x^+
=\rho_- n_t^-+\rho_- u_-\cdot n_x^-.
\]
Using $n_t^\pm=\mp v_s$ and $n_x^\pm=\pm n$, we obtain the Rankine--Hugoniot condition for mass \eqref{eq:RH_mass}.6

\medskip

\noindent\textit{Rankine–Hugoniot condition for momentum.}
Finally, we derive the conditions on the singular surface $\Gamma$. For any $\xi_\pm$ on $\Gamma$, we have
\[
0=\frac{\partial\mathscr l_\pm}{\partial u_\pm}(\xi_\pm)n_t
+ \frac{\partial\mathscr l}{\partial u_\pm}(\xi_\pm) \;u_\pm \cdot n_x^{\pm}
-\lambda_\pm \xi_\pm\cdot n_x
+\rho_\pm\xi_\pm\cdot\nabla w_\pm n_t^\pm
+\rho_\pm\xi_\pm\cdot\nabla w_\pm u_\pm\cdot n_x^\pm.
\]
Using the geometric identities
\[
n_t=-v_s,\qquad n_x= n,
\]
we arrive at
\[
v_s\frac{\partial\mathscr l_\pm}{\partial u_\pm}(\cdot)
=\left(\frac{\partial\mathscr l_\pm}{\partial u_\pm}\otimes u_\pm-\lambda_\pm\delta\right)\cdot n
-v_s\rho_\pm\nabla w_\pm(\cdot)
+\rho_\pm\nabla w_\pm(\cdot) u_\pm\cdot n.
\]
Subtracting these two equations, we obtain
$$
v_s\left[\!\left[\frac{\partial\mathscr l}{\partial u}\right]\!\right](\cdot)
=\left(\left[\!\left[\frac{\partial\mathscr l}{\partial u}\otimes u\right]\!\right]
+\left[\!\left[-\lambda\right]\!\right]\delta\right)\cdot n-(v_s\left[\!\left[\rho\right]\!\right]-\left[\!\left[\rho u\right]\!\right]\cdot n)\nabla w(\cdot).
$$

Then, for any $\xi^s=\xi_s\cdot n$ on $\Gamma$, we get
\begin{equation}\label{ell_lambda}
\mathscr l_++\lambda_++\rho_+ D_t w_+
=\mathscr l_-+\lambda_-+\rho_- D_t w_-,\quad \text{i.e.,}\quad\left[\!\left[-\lambda\right]\!\right]=\left[\!\left[\mathscr l +\rho D_tw\right]\!\right].
\end{equation}
Together with the Rankine--Hugoniot condition for mass, we obtain the Rankine--Hugoniot condition for momentum on $\Gamma$, given by Equation \eqref{genLagRH}. 
\end{proof}

\paragraph{Bulk, interface, and boundary equations in standard PDE notation.}
The system \eqref{genEulerLag} corresponds to equations of motion associated with the Lagrangian $\mathscr l (u_\pm,\rho_\pm)$ for the bulk fluid on each side of the interface $\Gamma$. The first equation,
\[
\partial_t \rho_\pm + \mathscr{L}_{u_\pm}\rho_\pm = 0,
\]
represents the continuity equation expressed in terms of the Lie derivative, which in standard vector calculus notation reads
\[
\partial_t \rho_\pm + \nabla\cdot(\rho_\pm u_\pm) = 0.
\]
The second equation,
\[
\partial_t \frac{\partial \mathscr l}{\partial u_\pm}
+ \mathscr{L}_{u_\pm}\frac{\partial \mathscr l}{\partial u_\pm}
= \rho_\pm d \frac{\partial \mathscr l}{\partial \rho_\pm},
\]
is the momentum equation derived variationally. Defining
\(
m_\pm := \frac{\partial \mathscr l}{\partial u_\pm},
\)
and using $\mathscr L_u m = \nabla_u m + \nabla u^T m + m \operatorname{div}u =  \nabla\cdot (m \otimes u)+\nabla u^T m$,
it can be written in the form
\begin{equation}\label{m_u_form}
\partial_t m_\pm
+ \nabla\cdot (m_\pm \otimes u_\pm)+\nabla u_\pm^T m_\pm
- \rho_\pm \nabla \frac{\partial \mathscr l}{\partial \rho_\pm} = 0.
\end{equation}
For the Lagrangian density $\mathscr l =\frac 12\rho|u|^2-\varepsilon(\rho)$, we get
$\nabla u_\pm^T m_\pm
- \rho_\pm \nabla \frac{\partial \mathscr l}{\partial \rho_\pm}=\rho_\pm \nabla (\varepsilon'(\rho_\pm))=\nabla p_\pm$, where
$$
p=\mathscr l -\rho\frac{\partial\mathscr l }{\partial\rho}=-\varepsilon+\rho\varepsilon'(\rho),
$$
is the pressure.
Hence \eqref{m_u_form} reduces to the standard Euler momentum equation.

Equation \eqref{genBounCon} for free-boundary fluids represents the natural boundary condition on $S_\pm$, which enforces the vanishing of the normal flux of the momentum on the boundary:
\[
\mathscr l_\pm - \rho_\pm \frac{\partial \mathscr l}{\partial \rho_\pm}= 0.
\]
For example, if $\mathscr l $ is the standard kinetic-minus-potential Lagrangian, this reduces to the pressure boundary condition $p_\pm = 0$ on $S_\pm$.

Equation \eqref{genLagRH} expresses the momentum jump across the interface $\Gamma$:
\[
v_s [\![m]\!]
= [\![m \otimes u]\!] n
+ [\![\mathscr l  - \rho \partial_\rho \mathscr l ]\!] n,
\]
where $v_s$ is the normal speed of the discontinuity and $[\![\cdot]\!]$ denotes the jump across $\Gamma$. This is the 
standard Rankine--Hugoniot condition ensuring conservation of momentum across the moving discontinuity.

Similarly, equation \eqref{eq:RH_mass} is the classical Rankine--Hugoniot condition for mass, stating that the normal mass flux across the interface balances the motion of the discontinuity:
\[
v_s [\![\rho]\!] = [\![\rho u]\!] \cdot  n.
\]

In summary, Theorem \ref{RHvar_modi} demonstrates that a single variational framework yields the bulk equations, boundary conditions, and Rankine–Hugoniot jump conditions, which in standard PDE notation correspond to the compressible barotropic Euler equations, associated boundary conditions, and classical mass and momentum jump conditions at the interface.

\section{Energy balance and dissipation: role of $\lambda_\pm$}\label{section_EB}

In this section, we interpret the potential term involving $\Lambda_\pm$ in the Lagrangian formulation (see \eqref{Lagrangian}), together with its Eulerian counterpart $\lambda_\pm = \Lambda_\pm \circ \varphi_\pm^{-1}, J\varphi_\pm^{-1}$ (see \eqref{critical_AP}), as a dissipation potential for barotropic shocks. The contribution of this potential is entirely localized on the shock surface.
This interpretation is obtained by deriving the energy balance for a barotropic fluid in the presence of shocks (\S\ref{dissip_pot}). We then illustrate the approach in the one-dimensional setting (\S\ref{1D_example}).

\subsection{Dissipation potential for barotropic shocks}\label{dissip_pot}

In the presence of shocks, the energy is no longer conserved, and its rate of change becomes a key quantity in the analysis. For each side of the shock, we define the total energy by
\[
\mathscr E_\pm(t)=\int_{M_\pm(t)}E_\pm\,{\rm d}x,
\qquad
\mathscr E(t)=\mathscr E_+(t)+\mathscr E_-(t),
\]
where the total energy density is defined from the Lagrangian density as
\[
E_\pm =\frac{\partial \mathscr l_\pm}{\partial u_\pm}\cdot u_\pm - \mathscr l_\pm.
\]

\begin{proposition}\label{prop:energy}
The time derivative of the total energy satisfies
\begin{equation}\label{time_change_E}
\frac{d}{dt}\mathscr E=\int_{\Gamma_t} \left(v_s\left[\!\left[E\right]\!\right]-\left[\!\left[(E+p)u\right]\!\right]\cdot n\right)\,{\rm d}a-\int_{\partial M} p u\cdot n\,{\rm d}a.
\end{equation}
\end{proposition}
\begin{proof}
We have
$$
\frac{d}{dt}\mathscr E_\pm=\int_{\partial M_\pm(t)} E_\pm u^s\cdot n_\pm\,{\rm d}a-\int_{\partial M_\pm(t)}(E_\pm+p_\pm)u_\pm\cdot n_\pm\,{\rm d}a.
$$
On the shock surface $\Gamma_t$ one has $u^s\cdot n_\pm=v_s$, while on $S_\pm$ the boundary velocity coincides with the fluid velocity. Summing the $+$ and $-$ contributions and using $n=n_+=-n_-$ on $\Gamma_t$ yields
$$
\frac{d}{dt}\mathscr E=\frac{d}{dt}\mathscr E_++\frac{d}{dt}\mathscr E_-=\int_{\Gamma_t} \left(v_s\left[\!\left[E\right]\!\right]-\left[\!\left[(E+p)u\right]\!\right]\cdot n\right)\,{\rm d}a-\int_{\partial M} p u\cdot n\,{\rm d}a,
$$
as desired.
\end{proof}

The result above follows only from the energy balance in the smooth regions and does not depend on the interface or boundary conditions.
We now relate the energy dissipation \eqref{time_change_E} to the variable $\lambda_\pm$ introduced in the variational formulation. This result uses the critical conditions involving $\lambda_\pm$ as derived in Theorem \ref{RHvar_modi}.

\begin{theorem}
The rate of change of the total energy can be written as
\begin{equation}\label{time_change_E_lambda}
\frac{d}{dt}\mathscr E=\int_{\Gamma_t} \left(v_s\left[\!\left[\lambda\right]\!\right]-\left[\!\left[\lambda u\right]\!\right]\cdot n\right)\,{\rm d}a.
\end{equation}
\end{theorem}
\begin{proof}
Recall the variational identities
\[
E=\frac{\partial\mathscr l}{\partial u}\cdot u-\mathscr l,
\qquad
p=\mathscr l-\rho\frac{\partial\mathscr l}{\partial\rho}.
\]
Moreover, the natural boundary condition on $\partial M$ implies $p=0$ there. Hence \eqref{time_change_E} reduces to
$$
\frac{d}{dt}\mathscr E=\int_{\Gamma_t} \left(v_s\left[\!\left[E\right]\!\right]-\left[\!\left[(E+p)u\right]\!\right]\cdot n\right)\,{\rm d}a=\int_{\Gamma_t} \left(v_s\left[\!\left[\frac{\partial \mathscr l}{\partial u}\cdot u-\mathscr l\right]\!\right]-\left[\!\left[\left(\frac{\partial \mathscr l}{\partial u}\cdot u- \rho\frac{\partial \mathscr l}{\partial \rho}\right)u\right]\!\right]\cdot n\right)\,{\rm d}a.
$$
Using the Rankine--Hugoniot relations derived in Theorem~\ref{RHvar_modi}, one has on $\Gamma_t$
$$
v_s\frac{\partial\mathscr l_\pm}{\partial u_\pm}(\cdot)=\left(\frac{\partial\mathscr l_\pm}{\partial u_\pm}\otimes u_\pm-\lambda_\pm\delta\right)\cdot n_x^\pm+\rho_\pm\nabla w_\pm(\cdot) n_t^\pm+\rho_\pm\nabla w_\pm(\cdot) u_\pm\cdot n_x^\pm.
$$
Evaluating this identity on $u_\pm$ gives
$$
v_s\frac{\partial\mathscr l_\pm}{\partial u_\pm}\cdot u_\pm
=\left(\frac{\partial\mathscr l_\pm}{\partial u_\pm}\cdot u_\pm-\lambda_\pm\right)u_\pm\cdot n_x^\pm
+\rho_\pm\nabla w_\pm\cdot u_\pm\, n_t^\pm
+\rho_\pm\nabla w_\pm\cdot u_\pm\, u_\pm\cdot n_x^\pm,
$$
and using $n_t^+=-v_s$ gives
$$
v_s\left[\!\left[\frac{\partial \mathscr l}{\partial u}\cdot u\right]\!\right]
=\left[\!\left[\left(\frac{\partial \mathscr l}{\partial u}\cdot u-\lambda\right)u\right]\!\right]\cdot n
-\left[\!\left[\rho u\cdot\nabla w\right]\!\right]v_s
+\left[\!\left[\rho u\cdot\nabla w\, u\right]\!\right]\cdot n.
$$
Substituting into the expression for $\frac{d}{dt}\mathscr E$,
$$
\begin{aligned}
\frac{d}{dt}\mathscr E
&=\int_{\Gamma_t}
\left(v_s\left[\!\left[\frac{\partial \mathscr l}{\partial u}\cdot u-\mathscr l\right]\!\right]
-\left[\!\left[\left(\frac{\partial \mathscr l}{\partial u}\cdot u
-\rho\frac{\partial \mathscr l}{\partial \rho}\right)u\right]\!\right]\cdot n\right)\,{\rm d}a\\
&=\int_{\Gamma_t}
\left(v_s\left[\!\left[-\mathscr l\right]\!\right]
+\left[\!\left[\rho\frac{\partial \mathscr l}{\partial \rho}u\right]\!\right]\cdot n
-\left[\!\left[\lambda u\right]\!\right]\cdot n
-\left[\!\left[\rho u\cdot\nabla w\right]\!\right]v_s
+\left[\!\left[\rho u\cdot\nabla w\, u\right]\!\right]\cdot n\right)\,{\rm d}a.
\end{aligned}
$$
Then, using the identity
$\left[\!\left[-\mathscr l\right]\!\right]
=\left[\!\left[\lambda+\rho D_t w\right]\!\right]$
on $\Gamma$, see \eqref{ell_lambda}, together with continuity of $\partial_t w$ and $\nabla_x w$ across $\Gamma_t$, all terms involving $w$ cancel:
$$
\begin{aligned}
&v_s\left[\!\left[\rho D_t w\right]\!\right]
+\left[\!\left[\rho\frac{\partial \mathscr l}{\partial \rho}u\right]\!\right]\cdot n
-\left[\!\left[\rho u\cdot\nabla w\right]\!\right]v_s
+\left[\!\left[\rho u\cdot\nabla w\, u\right]\!\right]\cdot n\\
&=\,v_s\left[\!\left[\rho \partial_t w\right]\!\right]
+\left[\!\left[\rho\frac{\partial \mathscr l}{\partial \rho}u\right]\!\right]\cdot n
+\left[\!\left[\rho u\cdot\nabla w\, u\right]\!\right]\cdot n\\
&=\,\partial_t w\left[\!\left[\rho u\right]\!\right]\cdot n
+\left[\!\left[\rho\frac{\partial \mathscr l}{\partial \rho}u\right]\!\right]\cdot n
+\left[\!\left[\rho u\cdot\nabla w\, u\right]\!\right]\cdot n\\
&=\,\left[\!\left[\rho D_t w\, u\right]\!\right]\cdot n
+\left[\!\left[\rho\frac{\partial \mathscr l}{\partial \rho}u\right]\!\right]\cdot n
=0.
\end{aligned}
$$
The remaining contribution is precisely \eqref{time_change_E_lambda}.
\end{proof}

\paragraph{Reinterpretation of the construction.} 
The relation \eqref{time_change_E_lambda} shows that the variables $\lambda_\pm$ encode the irreversible energy dissipation occurring at the shock. In particular, $\lambda$ does \textit{not} satisfy a conservative Rankine–Hugoniot condition:
\[
v_s\big[\!\big[\lambda\big]\!\big]\neq\big[\!\big[\lambda u\big]\!\big]\cdot n .
\]
Furthermore, the contribution associated with $\lambda_\pm$ is localized entirely on the shock surface. Accordingly, in the previously introduced variational principle,
\[
\begin{aligned}
&\delta \int_{t_0}^{t_1} \left[\int_{\varphi_-(t)^{-1}(M_-(t))}\left(\mathcal{L}(\varphi_-,\dot{\varphi}_-,\nabla \varphi_-,\varrho_-) + \varrho_- \dot W \right){\rm d}X + \int_{\varphi_-(t)^{-1}(M_-(t))}\Lambda_-(X)\;{\rm d}X\right.\\
&\qquad\qquad +\left. \int_{\varphi_+(t)^{-1}(M_+(t))}\left(\mathcal{L}(\varphi_+,\dot{\varphi}_+,\nabla \varphi_+,\varrho_+) + \varrho_+ \dot W \right){\rm d}X + \int_{\varphi_+(t)^{-1}(M_+(t))}\Lambda_+(X)\;{\rm d}X\right]{\rm d}t=0,
\end{aligned},
\]
the terms involving $\Lambda_\pm$ can be naturally interpreted as defining a \textit{dissipation potential} added to the Lagrangian:
\begin{equation}\label{dissipation_potential}
-\mathcal{V}_{\rm shock}(\varphi_\pm, M_\pm) = \int_{\varphi_-(t)^{-1}(M_-(t))}\Lambda_-(X)\;{\rm d}X + \int_{\varphi_+(t)^{-1}(M_+(t))}\Lambda_+(X)\;{\rm d}X.
\end{equation}
The corresponding Eulerian version reads
\[
-\mathcal{V}_{\rm shock}(\varphi_\pm, M_\pm)= \int_{M_+} \Lambda_+\circ\varphi_+^{-1}J\varphi_+^{-1} \,{\rm d}x + \int_{M_-} \Lambda_-\circ\varphi_-^{-1}J\varphi_-^{-1} \,{\rm d}x.
\]
Thus, the functional $\mathcal V_{\rm shock}: Q_{\rm free}\text{(or $Q_{\rm fix}$)}\rightarrow\mathbb{R}$ represents the energy dissipated at the shock. This interpretation is further clarified by the following characterization of its gradient, which follows by direct computation:

\begin{lemma}\label{gradient_V} The functional $\mathcal{V}_{\rm shock}$ satisfies
\[
- \nabla_{\varphi_\pm, M_\pm}\mathcal{V}_{\rm shock} \cdot (\delta \varphi_\pm, \delta M_\pm)%= f_{\rm shock}\cdot (\xi_\pm, \xi_s)
= \int_{\Gamma _t} \left(\xi^s \left[\!\left[\lambda\right]\!\right]-\left[\!\left[\lambda\xi\right]\!\right]\cdot n\right)\,{\rm d}a,
\]
where $\xi_\pm = \delta\varphi_\pm\circ \varphi_\pm^{-1}$ and $\xi^s=\delta\Gamma= \xi_s\cdot n_+$.
\end{lemma}

\paragraph{Augmented energy.} Formally, the above construction fits into the abstract variational framework
\[
\delta \int_0^T\Big( L(q, \dot q) - \mathcal{V}(q) \Big) \,{\rm d}t=0
\]
for which the physical energy
\[
E(q,\dot q)=\left\langle\frac{\partial L}{\partial\dot q},\dot q\right\rangle-L(q,\dot q)
\]
satisfies the dissipation law
\begin{equation}\label{dissipation_E}
\frac{d}{dt}E =- \frac{d}{dt}\mathcal{V} = - \left\langle \frac{\partial\mathcal{V}}{\partial q}, \dot q \right\rangle=\left\langle F_{\rm shock} ,\dot q\right\rangle ,
\end{equation}
where $F_{\rm shock}=-\partial\mathcal V/\partial q$. On the other hand, the \textit{augmented} energy
\[
\widehat{E}(q, \dot q) =  \left\langle\frac{\partial L}{\partial\dot q},\dot q\right\rangle - L(q, \dot q) + \mathcal{V}(q)
\]
is formally conserved
\begin{equation}\label{E_aug}
\frac{d}{dt}\widehat{E}= 0 = \frac{d}{dt}E + \frac{d}{dt}\mathcal{V}.
\end{equation}
Equation \eqref{dissipation_E} is the abstract counterpart of \eqref{time_change_E_lambda}.
In particular, the right-hand side of \eqref{time_change_E_lambda} corresponds to the action of the gradient of $\mathcal{V}_{\rm shock}$ (modulo the use of its Eulerian representation $f_{\rm shock}$):
\[
- \nabla_{\varphi_\pm, M_\pm}\mathcal{V}_{\rm shock} \cdot (\delta \varphi_\pm, \delta M_\pm)= f_{\rm shock}\cdot (\xi_\pm, \xi^s) = \int_{\Gamma _t} \left(\xi^s \left[\!\left[\lambda\right]\!\right]-\left[\!\left[\lambda\xi\right]\!\right]\cdot n\right)\,{\rm d}a,
\]
consistently with Lemma \ref{gradient_V}. On the other hand, the conservation of the augmented energy in \eqref{E_aug} corresponds to the fact that the quantity $(E-\lambda)_\pm$ is conserved in the smooth regions (since $E_\pm$ and $\lambda_\pm$ are), while they satisfy a conservative Rankine--Hugoniot condition
\[
v_s\big[\!\big[E-\lambda\big]\!\big]=\big[\!\big[(E-\lambda) u\big]\!\big]\cdot n .
\]
across the shock surface.

\begin{remark} 
Choosing $\Lambda(X)=1$ yields $\lambda_\pm = J\varphi_\pm^{-1}$. It is natural to ask whether such a trivial constant choice for $\Lambda$ is admissible. While this choice is not essential—since the resulting Euler–Lagrange equations are independent of $\Lambda$—it does affect the specific values of $\lambda_\pm$.
To clarify this point, we examine explicit shock solutions of the one-dimensional barotropic Euler equations, for which the energy dissipation can be computed directly.
\end{remark}

\subsection{One-dimensional compressible fluids}\label{1D_example}

In conservative form, the one-dimensional barotropic Euler equations are
\begin{equation}\label{eq:conservative}
\begin{cases}
\partial_t \rho + \partial_x(\rho u) = 0, \\[4pt]
\partial_t(\rho u) + \partial_x\bigl(\rho u^2 + p(\rho)\bigr) = 0.
\end{cases}
\end{equation}
We consider the polytropic pressure law $p(\rho)=K\rho^\gamma$ with $\gamma>1$.

\medskip

Let a discontinuity propagate with speed $v_s$. The Rankine--Hugoniot conditions
are
\[
\begin{aligned}
v_s\left[\!\left[\rho\right]\!\right] &= \left[\!\left[\rho u\right]\!\right], \\
v_s\left[\!\left[\rho u\right]\!\right] &= \left[\!\left[\rho u^2 + p(\rho)\right]\!\right].
\end{aligned}
\]
%where $\left[\!\left[f\right]\!\right] = f_R - f_L$.
In particular,
\[
v_s = \frac{\rho_R u_R - \rho_L u_L}{\rho_R - \rho_L},
\]
and the Hugoniot relation can be written as
\[
(\rho_R u_R^2 + p(\rho_R)) - (\rho_L u_L^2 + p(\rho_L))
= v_s(\rho_R u_R - \rho_L u_L).
\]

\begin{remark}(Entropy pair)
A convex entropy--entropy flux pair $(E,q)$ is given by
\[
E(\rho,u)=\frac12\rho u^2+\rho e(\rho),
\qquad
e'(\rho)=\frac{p(\rho)}{\rho^2},
\]
with associated flux
\[
q
=u\left(E+p\right)=
u\bigl(\tfrac12\rho u^2+\rho e(\rho)+p(\rho)\bigr).
\]
For weak solutions one has the entropy inequality
\[
\partial_t E+\partial_x q\le0,
\]
which implies $v_s\left[\!\left[E\right]\!\right]-\left[\!\left[(E+p)u\right]\!\right]\cdot n\leq 0$ and therefore $\frac{d}{dt}\mathscr E\le 0$. This condition is imposed to select the physically admissible solution among all weak solutions that satisfy the Rankine--Hugoniot jump conditions.
\end{remark}

\begin{example}
(A special shock solution)
From now on, we consider the barotropic Euler system with polytropic pressure
$p(\rho)=K\rho^\gamma$, $K>0$, $\gamma>1$, and it is easy to check that the piecewise constant states
\begin{equation}\label{choice}
(\rho_L,u_L)=(1,2),\qquad(\rho_R,u_R)=(2,1),
\end{equation}
satisfies the one-dimensional barotropic Euler equations \eqref{eq:conservative}.
These satisfy $\rho_R>\rho_L$ and $u_L>u_R$, corresponding to a compressive
shock.

\paragraph{Shock speed.}
From the mass jump condition, \eqref{choice} yields
\[
v_s=\frac{\rho_Ru_R-\rho_Lu_L}{\rho_R-\rho_L}
%=\frac{2\cdot1-1\cdot2}{2-1}
=0,
\]
so the shock is stationary.

\paragraph{Momentum jump condition.}
The second Rankine--Hugoniot condition becomes
\[
0=(2\cdot1^2+K2^\gamma)-(1\cdot2^2+K),
\]
that is,
$2+K2^\gamma=4+K$.
Solving for $K$ gives
\[
K=\frac{2}{2^\gamma-1}.
\]
With this choice, the exact weak solution is
\[
(\rho,u)(t,x)=
\begin{cases}
(1,2),&x<0,\\[4pt]
(2,1),&x>0,
\end{cases}
\]
representing a stationary shock at $x=0$.
\end{example}

\paragraph{Energy balance for the special solution.}
For barotropic Euler with $p(\rho)=K\rho^\gamma$,
\[
e(\rho)=\frac{K}{\gamma-1}\rho^{\gamma-1},
\qquad
E(\rho,u)=\frac12\rho u^2+\frac{K}{\gamma-1}\rho^\gamma.
\]
For the left and right states \((\rho_L,u_L)=(1,2)\), \((\rho_R,u_R)=(2,1)\), the left and right energy densities are
\[
E_L
= \frac12\cdot 1\cdot 2^2+\frac{K}{\gamma-1}
= 2+\frac{K}{\gamma-1},
\]

\[
E_R
= \frac12\cdot 2\cdot 1^2+\frac{K}{\gamma-1}2^\gamma
= 1+\frac{K}{\gamma-1}2^\gamma.
\]
Since $v_s=0$, using \eqref{time_change_E}, the time derivative of the total energy
$\mathscr E=\int_M E\,{\rm d}x$ is
\[
\frac{d}{dt}\mathscr{E}= (E_R+p_R)  u_R - (E_L+p_L)  u_L
\]

A direct computation gives
\[
(E_R+p_R)- 2(E_L+p_L)
=
-3+\frac{K\gamma}{\gamma-1}(2^\gamma-2).
\]
Substituting $K=\dfrac{2}{2^\gamma-1}$ yields
\[
(E_R+p_R)- 2(E_L+p_L)
=
-3+\frac{2\gamma}{\gamma-1}
-\frac{2\gamma}{(\gamma-1)(2^\gamma-1)}.
\]

Meanwhile, the time derivative of volume (in one dimension, the change of length) is $u_R-u_L$.
%\[
%(u_L - u_R)\,\Delta t = \Delta t .
%\]
This can be justified as follows. Consider
\[
\frac{d}{dt} \left(
\int_{\varphi_-(t)^{-1}(M_-(t))} \mathrm{d}X
+
\int_{\varphi_+(t)^{-1}(M_+(t))} \mathrm{d}X
\right).
\]
Passing to Eulerian variables, this becomes
\[
\frac{d}{dt}\left(
\int_{M_-(t)} \lambda_-\, \mathrm{d}x
+
\int_{M_+(t)} \lambda_+\, \mathrm{d}x
\right),
\]
where $\lambda_\pm = J\varphi_\pm$ denotes the Jacobian determinant of the flow map on each side.

Using the transport theorem for a moving interface $\Gamma$, we obtain
\[
\frac{d}{dt}
\left(
\int_{M_-(t)} \lambda_-\, \mathrm{d}x
+
\int_{M_+(t)} \lambda_+\, \mathrm{d}x
\right)
=
\int_{\Gamma}
\Big(
- v_s \llbracket \lambda \rrbracket
+
\llbracket \lambda u \rrbracket \cdot n
\Big)\,{\rm d}a .
\]
In the present situation, the motion of the specific shock solution \eqref{choice} consists simply of uniform translation with constant velocity on each side.

Therefore,
\[
\lambda_\pm = J\varphi_\pm = 1,
\]
so there is no local volumetric deformation. The only contribution to the change of length comes from the jump in velocity across the interface.

Consequently, plugging $v_s=0$ and $\lambda_\pm = 1$ into the above equation, we get the time derivative of length is $u_R-u_L=-1$.

\paragraph{Conclusion.} This example shows that, for general $\gamma$,
\[
\frac{d}{dt}\mathscr E
\neq
-\frac{d}{dt}\mathcal V_{\rm shock},
\]
when the dissipation potential is chosen to be the volume functional (corresponding to $\Lambda_\pm(X)=1$)
$$
-\mathcal{V}_{\rm shock}(\varphi_\pm, M_\pm) = \int_{\varphi_-(t)^{-1}(M_-(t))}{\rm d}X + \int_{\varphi_+(t)^{-1}(M_+(t))}{\rm d}X.
$$
Consequently, the term $\lambda$ appearing in the modified Lagrangian
$\mathscr l +\lambda$ cannot, in general, be interpreted as arising from the volume
functional alone.

\section{Variational formulation for the compressible Euler system with entropy and shocks}\label{sec_5}

In the variational framework developed above, the presence of shocks leads to a loss of total energy. While this is appropriate for the barotropic Euler equations, where no entropy variable is present, it is no longer adequate for the full compressible Euler system. In that case, total energy is conserved, and the effects of shocks are instead reflected in entropy production.

This shift places the problem within the framework of nonequilibrium thermodynamics and calls for an extension of the variational formulation to incorporate irreversible processes in a manner consistent with energy conservation.

\subsection{A variational strategy for entropy production with energy conservation}

Consistently with the above considerations, we shall use a variational formulation for nonequilibrium thermodynamics (see \cite{GBYo2019}), which is able to capture the mechanism by which energy dissipated at shocks is not lost but converted into internal energy, resulting in an increase of entropy.

\paragraph{Variational formulation with irreversible processes.} 

Abstractly, such a variational principle can be seen as an extension of Hamilton’s principle and takes the form of a constrained variational principle
\[
\delta \int_{t_0}^{t_1} L(q, \dot q, S) + (S - \Sigma)\,\dot{\Gamma}\, {\rm d}t = 0,
\]
subject to the phenomenological and variational constraints
\begin{equation}\label{constraints}
\frac{\partial L}{\partial S}\,\dot \Sigma
=
\langle F, \dot q\rangle,\qquad \frac{\partial L}{\partial S}\,\delta \Sigma
=
\langle F, \delta q\rangle,
\end{equation}
for some dissipative force $F$.

This principle yields the following conditions:
\[
\frac{d}{dt}\frac{\partial L}{\partial \dot q}
- \frac{\partial L}{\partial q}
= F, 
\qquad
\frac{\partial L}{\partial S}\,\dot S
= \langle F, \dot q\rangle,
\qquad
\dot S = \dot \Sigma,
\qquad
\dot \Gamma = - \frac{\partial L}{\partial S}.
\]

The first two equations correspond to the momentum and entropy balances, while the last two determine the auxiliary variables $\dot \Gamma$ and $\dot \Sigma$. This type of principle will be implemented for the compressible fluid with shocks below.

\begin{remark}\label{remark_dAlembert}[D'Alembert-type principle]
It is important to emphasize that the above principle involves two distinct types of constraints. The first restricts the critical curves of the action functional, while the second restricts the admissible class of variations in the action principle. These two constraints are systematically related, in complete analogy with d'Alembert's principle, through the fact that the variations (generalized virtual displacements, here $\delta \Sigma$ and $\delta q$) and the time derivatives (here $\dot{\Sigma}$ and $\dot{q}$) satisfy the same constraint.
\end{remark}

\begin{remark}[Role of $\Sigma$] In the simple setting considered above, the variable $\Sigma$ may appear superfluous. Indeed, replacing $\Sigma$ by $S$ in the constraints and omitting the term $(S-\Sigma)\dot{\Gamma}$ in the action functional leads to the same final dynamical equations. This simplification, however, does not hold in general (see, for instance, \S\ref{sec:entropy}), and $\Sigma$ plays a crucial role in nonequilibrium thermodynamics. The key point is that $\dot{S}$ and $\dot{\Sigma}$ have distinct physical meanings: $\dot{S}$ is the rate of change of the system entropy, whereas $\dot{\Sigma}$ represents the entropy production rate, which must be nonnegative by the second law of thermodynamics.
Although these quantities coincide in the present example, their distinction is essential in general (see, e.g., \cite{GBYo2019,VDPGBYoCh2026}).
\end{remark}

\begin{remark}[Case of fluids]
To gain intuition on how the principle applies to fluids, observe that the additional term $(S - \Sigma)\dot{\Gamma}$ in the action functional retains the same form in the material description, while in the Eulerian description it becomes
\[
(s - \varsigma) D_t \gamma,
\]
where $s$, $\varsigma$, and $\gamma$ denote the Eulerian fields corresponding to $S$, $\Sigma$, and $\Gamma$.
When viscous processes are considered in a nonequilibrium thermodynamic setting, the constraints in \eqref{constraints} take the form
\[
\frac{\partial \mathscr l }{\partial s} \, D_t \varsigma
=
- \sigma : \nabla u,
\qquad
\frac{\partial \mathscr l }{\partial s} \, D_\delta \varsigma
=
- \sigma : \nabla \xi,
\]
where $\sigma$ is the viscous stress tensor and
\begin{equation}\label{mat_var}
D_\delta = \delta  + \xi \cdot \nabla
\end{equation}
denotes the Eulerian variation.
These relations express that mechanical dissipation, given by $\sigma : \nabla u$, is converted into entropy production, thereby ensuring conservation of total energy.

In the present shock-driven setting, however, dissipation does not arise from a bulk contribution such as $\sigma : \nabla u$, but instead from singular contributions supported on the shock surface. Moreover, the entropy production at the shock is not determined by phenomenological laws and therefore cannot be prescribed a priori. Accordingly, the constraints reduce to the simpler form
\[
\bar D_t \varsigma = 0,
\qquad
\bar D_\delta \varsigma = 0,
\]
which serve as bookkeeping conditions ensuring that the energy dissipated at the shock is exactly re-injected into the system through entropy production.
\end{remark}

\paragraph{Action functional.} Following the above considerations and those in \S\ref{action_functional}, the action principle reads as follows:
\begin{equation}\label{Lagrangian_thermo}
\begin{aligned}
&\delta\left[\int_{t_0}^{t_1}
\int_{\varphi_-^{-1}(M_-(t))}
\Big( \mathcal{L}(\varphi_-, \dot{\varphi}_-, \nabla \varphi_-, \varrho_-) + \varrho_- \dot{W}  + (S_--\Sigma_-)\dot\Gamma_-\Big)\, \mathrm{d}X\right. \\
&\qquad\qquad  \left.+ \int_{\varphi_+^{-1}(M_+(t))}
\Big( \mathcal{L}(\varphi_+, \dot{\varphi}_+, \nabla \varphi_+, \varrho_+) + \varrho_+ \dot{W} + (S_+-\Sigma_+)\dot\Gamma_+\Big)\, \mathrm{d}X\right]{\rm d}t=0,
\end{aligned}
\end{equation}
subject to the constraints
\[
\frac{d}{dt}\Sigma=0,\qquad \delta \Sigma=0.
\]

In the Eulerian description, the resulting critical action principle is
\begin{equation}\label{VP_thermo1}
\begin{aligned}
\delta \int_{t_0}^{t_1}
\Bigg[
&\int_{M_-(t)}
\left(
\mathscr l (u_-,\rho_-,s_-)
+ \rho_- D_t w_-
+ (s_- - \varsigma_-) D_t \gamma_-
\right)\,{\rm d}x
\\
&+
\int_{M_+(t)}
\left(
\mathscr l (u_+,\rho_+,s_+)
+ \rho_+ D_t w_+
+ (s_+ - \varsigma_+) D_t \gamma_+
\right)\,{\rm d}x
\Bigg] dt = 0,
\end{aligned}
\end{equation}
subject to kinematic and variational constraints
\begin{equation}\label{constraint_thermo1}
\bar D_t \varsigma = 0,
\qquad
\bar D_\delta \varsigma = 0,
\end{equation}
for arbitrary variations $\delta \Gamma$, $\delta \rho_\pm$, and $\delta w_\pm$,
together with velocity variations of the form
\[
\delta u_\pm = \partial_t \xi_\pm + \mathscr{L}_{u_\pm}\xi_\pm.
\]
As earlier, $\xi_\pm = \delta\varphi_\pm \circ \varphi_\pm^{-1}$ are arbitrary vector
fields, and the variations of the shock hypersurface $\Gamma$ are of the form
$\delta\Gamma = \xi_s \cdot n$, where $\xi_s$ is an arbitrary vector field defined on $\Gamma$. We recall that these variational principles are of d’Alembert type (see Remark \ref{remark_dAlembert}).

\subsection{The compressible Euler equations and Rankine--Hugoniot conditions}\label{sec:compr_Euler_RH}

In the following theorem, we show how the bulk equations, boundary conditions, and Rankine--Hugoniot jump conditions arise from the variational formulation. The derivation of the Rankine--Hugoniot condition for energy is deferred to the next section.

\begin{theorem}\label{RHvar_rho_s}
The critical action principle \eqref{VP_thermo1}-\eqref{constraint_thermo1} yields the following equations:
\begin{itemize}
\item[(1)] \textbf{Bulk equations.} The fluid equations in the bulk are
\begin{equation}\label{genEulerLag_rho_s}
\left\{
\begin{aligned}
&\partial_t \rho_\pm+\mathscr L_{u_\pm}\rho_\pm=0,\\
&\partial_t s_\pm+\mathscr L_{u_\pm}s_\pm=0,\\
&\partial_t \frac{\partial\mathscr l}{\partial u_\pm}
+\mathscr L_{u_\pm}\frac{\partial\mathscr l}{\partial u_\pm}
=\rho_\pm d\frac{\partial\mathscr l}{\partial \rho_\pm}
+s_\pm d\frac{\partial\mathscr l}{\partial s_\pm}.
\end{aligned}
\right.
\end{equation}

\item[(2)] \textbf{Boundary conditions.} 
For free-boundary fluids, the conditions on the boundaries $S_\pm$ are
\begin{equation}\label{genBounCon_rho_s}
\mathscr l(u_\pm,\rho_\pm,s_\pm)
-\rho_\pm\frac{\partial\mathscr l}{\partial \rho_\pm}
-s_\pm\frac{\partial\mathscr l}{\partial s_\pm}=0.
\end{equation}
For fluids in a fixed domain, no additional boundary conditions arise. In particular, the condition $u_\pm \cdot n_\pm = 0$ on $S_\pm$, where $n_\pm$ denotes the outward unit normal vector to $S_\pm$, follows directly from the choice of configuration manifold.
\item[(3)] \textbf{Rankine--Hugoniot condition for momentum.}
On the singular surface $\Gamma$, one has
\begin{equation}\label{genLagRH_rho_s}
v_s\left[\!\left[\frac{\partial\mathscr l}{\partial u}\right]\!\right](\cdot)
=
\left(
\left[\!\left[\frac{\partial\mathscr l}{\partial u}\otimes u\right]\!\right]
+\left[\!\left[\mathscr l
-\rho\frac{\partial\mathscr l}{\partial \rho}
-s\frac{\partial\mathscr l}{\partial s}\right]\!\right]\delta
\right)\cdot n,
\end{equation}
where $v_s$ denotes the normal speed of $\Gamma$, $\delta$ is the $(1,1)$ identity tensor characterized by $\delta^i_j = 1$ if $i=j$ and $\delta^i_j=0$ otherwise, and $n$ is the unit normal vector field along $\Gamma$.

\item[(4)] \textbf{Rankine--Hugoniot condition for mass.}
On the singular surface $\Gamma$, one has
\begin{equation}\label{eq:RH_mass_rho_s}
v_s\left[\!\left[\rho\right]\!\right]=\left[\!\left[\rho u\right]\!\right]\cdot n.
\end{equation}
\item[(5)] \textbf{Rankine--Hugoniot condition for $(s-\varsigma)$.}
On the singular surface $\Gamma$, one has
\begin{equation}\label{RH_entropy_rho_s}
v_s\left[\!\left[s-\varsigma\right]\!\right]
=\left[\!\left[(s-\varsigma)u\right]\!\right]\cdot n.
\end{equation}
\end{itemize}
\end{theorem}
\begin{proof}
For arbitrary variations $q_\epsilon$ of the curve $q(t) \in \mathcal{Q}_{\rm free}$ (or $\mathcal{Q}_{\rm fix}$), see \eqref{q_variables}, with fixed endpoints, the critical action principle is given by
\begin{equation}\label{stan_eq_rho_s}
\begin{aligned}
0=&\int_{t_0}^{t_1}
\left[
    \int_{M_-(t)}
        \left( \mathscr l (u_-,\rho_-,s_-)
        + \rho_- D_t w_-
        + (s_- - \varsigma_-) D_t \gamma_- \right)\,{\rm d}x
\right. \\
&\left.
   + \int_{M_+(t)}
        \left( \mathscr l (u_+,\rho_+,s_+)
        + \rho_+ D_t w_+
        + (s_+ - \varsigma_+) D_t \gamma_+ \right)\, {\rm d}x
\right]\,{\rm d}t\\
=&\int_{t_0}^{t_1}{\rm d}t\int_{M_+}\left(\frac{\partial\mathscr l}{\partial u_+}\cdot\delta u_+ +\frac{\partial\mathscr l}{\partial \rho_+}\cdot\delta \rho_+ +\frac{\partial\mathscr l}{\partial s_+}\cdot\delta s_++\delta\rho_+D_t w_++\rho_+\delta D_tw_++\delta(s_+ - \varsigma_+)D_t\gamma_++(s_+ - \varsigma_+)\delta D_t\gamma_+\right){\rm d}x\\
+&\int_{t_0}^{t_1}{\rm d}t\left(\int_{S_+}(\mathscr l (u_+,\rho_+,s_+)+ \rho_+ D_t w_++ (s_+ - \varsigma_+) D_t \gamma_+)\xi_+\cdot n_+{\rm d}a+\int_{\Gamma}(\mathscr l (u_+,\rho_+,s_+)+ \rho_+ D_t w_++ (s_+ - \varsigma_+) D_t \gamma_+)\xi_s\cdot n_+{\rm d}a\right)\\
+&\int_{t_0}^{t_1}{\rm d}t\int_{M_-}\left(\frac{\partial\mathscr l}{\partial u_-}\cdot\delta u_- +\frac{\partial\mathscr l}{\partial \rho_-}\cdot\delta \rho_- +\frac{\partial\mathscr l}{\partial s_-}\cdot\delta s_-+\delta\rho_-D_t w_-+\rho_-\delta D_tw_-+\delta(s_- - \varsigma_-)D_t\gamma_-+(s_- - \varsigma_-)\delta D_t\gamma_-\right){\rm d} x\\
+&\int_{t_0}^{t_1}{\rm d}t\left(\int_{S_-}(\mathscr l (u_-,\rho_-,s_-)+ \rho_- D_t w_-+ (s_- - \varsigma_-) D_t \gamma_-)\xi_-\cdot n_-{\rm d}a+\int_{\Gamma}(\mathscr l (u_-,\rho_-,s_-)+ \rho_- D_t w_-+ (s_- - \varsigma_-) D_t \gamma_-)\xi_s\cdot n_-{\rm d}a\right).\\
\end{aligned}
\end{equation}
Note that if the domain is fixed, then on $S_\pm$, we have $\xi_\pm\cdot n_\pm=0$.

We consider a spacetime domain $\{M(t):t\in[t_0,t_1]\}=\Omega_+\cup\Omega_-$, and $\Sigma$ a spacetime singular submanifold as shown in Figure \ref{fig:shock}. We recall that $\delta u_\pm=\partial_t\xi_\pm+[u_\pm,\xi_\pm],$
where $\xi_\pm = \delta\varphi_\pm \circ \varphi_\pm^{-1}$
is an arbitrary time dependent vector field vanishing
at $t = t_0,\, t_1$. Hence, we have
\begin{equation}\label{first_eq_rho_s}
\begin{aligned}
&\iint_{\Omega_\pm}\left(\frac{\partial\mathscr l}{\partial u_\pm}\cdot\delta u_\pm +\frac{\partial\mathscr l}{\partial \rho_\pm}\cdot\delta \rho_\pm +\frac{\partial\mathscr l}{\partial s_\pm}\cdot\delta s_\pm+\delta\rho_\pm D_t w_\pm +\rho_\pm\delta D_tw_\pm+\delta(s_\pm - \varsigma_\pm)D_t\gamma_\pm +(s_\pm - \varsigma_\pm)\delta D_t\gamma_\pm\right)\;{\rm d}x{\rm d}t\\
=&\iint_{\Omega_\pm}\left(\frac{\partial\mathscr l}{\partial u_\pm}\cdot (\partial_t\xi_\pm+[u_\pm,\xi_\pm]) +\left(\frac{\partial\mathscr l}{\partial \rho_\pm}+D_t w_\pm\right)\delta\rho_\pm +\left(\frac{\partial\mathscr l}{\partial s_\pm}+D_t \gamma_\pm\right)\delta s_\pm - D_t\gamma_\pm \delta\varsigma_\pm\right)\;{\rm d}x{\rm d}t\\
&+\iint_{\Omega_\pm}\left(\rho_\pm\left(D_t\delta w_\pm+(\partial_t \xi_\pm+\mathscr L_{u_\pm}\xi_\pm,\nabla) w_\pm\right)+(s_\pm - \varsigma_\pm)\left(D_t\delta \gamma_\pm+(\partial_t \xi_\pm+\mathscr L_{u_\pm}\xi_\pm,\nabla) \gamma_\pm\right)\right)\;{\rm d}x{\rm d}t\\
=&\iint_{\Omega_\pm}\left(\left(-\partial_t \frac{\partial\mathscr l}{\partial u_\pm}-\mathscr L_{u_\pm}\frac{\partial\mathscr l}{\partial u_\pm}\right)\cdot\xi_\pm+\left(\frac{\partial\mathscr l}{\partial \rho_\pm}+D_t w_\pm\right)\delta\rho_\pm+\left(\frac{\partial\mathscr l}{\partial s_\pm}+D_t \gamma_\pm\right)\delta s_\pm
-\varsigma_\pm \nabla D_t\gamma_\pm\cdot\xi_\pm\right)\;{\rm d}x{\rm d}t\\
&+\iint_{\Omega_\pm}\left(\bar{D}_t(\rho_\pm\delta w_\pm)-\bar{D}_t\rho_\pm\delta w_\pm+\bar{D}_t(\rho_\pm\xi_\pm\cdot\nabla w_\pm)-\xi_\pm\cdot(\partial_t(\rho_\pm\nabla w_\pm)+\mathscr L_{u_\pm} (\rho_\pm\nabla w_\pm))\right)\;{\rm d}x{\rm d}t\\
&+\iint_{\Omega_\pm}\left(\bar{D}_t((s_\pm-\varsigma_\pm)\delta \gamma_\pm)-\bar{D}_t(s_\pm-\varsigma_\pm)\delta \gamma_\pm+\bar{D}_t((s_\pm-\varsigma_\pm)\xi_\pm\cdot\nabla \gamma_\pm)-\xi_\pm\cdot(\partial_t((s_\pm-\varsigma_\pm)\nabla \gamma_\pm)+\mathscr L_{u_\pm} ((s_\pm-\varsigma_\pm)\nabla \gamma_\pm))\right)\;{\rm d}x{\rm d}t\\
+&\int_{\partial\Omega_\pm}\left(\frac{\partial\mathscr l}{\partial u_\pm}\cdot \xi_\pm n_t^\pm + \frac{\partial\mathscr l}{\partial u_\pm}(\xi_\pm) u_\pm \cdot n_x^{\pm}+\varsigma_\pm D_t\gamma_\pm \xi_\pm\cdot n_x^\pm\right)\;{\rm d}a(x,t)\\
=&\iint_{\Omega_\pm}\left(-\partial_t \frac{\partial\mathscr l}{\partial u_\pm}-\mathscr L_{u_\pm}\frac{\partial\mathscr l}{\partial u_\pm}-\varsigma_\pm \nabla D_t\gamma_\pm-=(\partial_t(\rho_\pm\nabla w_\pm)+\mathscr L_{u_\pm} (\rho_\pm\nabla w_\pm))-(\partial_t((s_\pm-\varsigma_\pm)\nabla \gamma_\pm)+\mathscr L_{u_\pm} ((s_\pm-\varsigma_\pm)\nabla \gamma_\pm))\right)\cdot\xi_\pm\;{\rm d}x{\rm d}t\\
&+\iint_{\Omega_\pm}\left(\left(\frac{\partial\mathscr l}{\partial \rho_\pm}+D_t w_\pm\right)\delta\rho_\pm-\bar{D}_t\rho_\pm\delta w_\pm+\left(\frac{\partial\mathscr l}{\partial s_\pm}+D_t \gamma_\pm\right)\delta s_\pm-\left(\bar{D}_t(s_\pm-\varsigma_\pm)\right)\delta \gamma_\pm\right)\;{\rm d}x{\rm d}t\\
+&\int_{\partial\Omega_\pm}\left(\frac{\partial\mathscr l}{\partial u_\pm}\cdot \xi_\pm n_t^\pm + \frac{\partial\mathscr l}{\partial u_\pm}(\xi_\pm) u_\pm \cdot n_x^{\pm}+\varsigma_\pm D_t\gamma_\pm \xi_\pm\cdot n_x^\pm\right)\;{\rm d}a(x,t)\\
&+\int_{\partial\Omega_\pm} \left(\rho_\pm\delta w_\pm n_t^\pm+\rho_\pm\delta w_\pm u_\pm\cdot n_x^\pm+\rho_\pm\xi_\pm\cdot\nabla w_\pm n_t^\pm+\rho_\pm\xi_\pm\cdot\nabla w_\pm u_\pm\cdot n_x^\pm\right)\;{\rm d}a(x,t)\\
&+\int_{\partial\Omega_\pm} \left((s_\pm-\varsigma_\pm)\delta \gamma_\pm n_t^\pm+(s_\pm-\varsigma_\pm)\delta \gamma_\pm u_\pm\cdot n_x^\pm+(s_\pm-\varsigma_\pm)\xi_\pm\cdot\nabla \gamma_\pm n_t^\pm+(s_\pm-\varsigma_\pm)\xi_\pm\cdot\nabla \gamma_\pm u_\pm\cdot n_x^\pm\right)\;{\rm d}a(x,t),
\end{aligned}
\end{equation}
where we used the formulas \eqref{useful_formulas} and the variational constraint
$$
-D_t\gamma\delta\varsigma=D_t\gamma\;{\rm div}(\varsigma\xi)
=-\varsigma \nabla D_t\gamma\cdot\xi+{\rm div}\left(\varsigma D_t\gamma\xi\right).
$$

\noindent\textit{Bulk equations.} Plugging \eqref{first_eq_rho_s} into \eqref{stan_eq_rho_s}, we get in $\Omega_\pm$:
\begin{align*}
\xi_\pm:&\;\;\;
\partial_t \frac{\partial\mathscr l}{\partial u_\pm}+\mathscr L_{u_\pm}\frac{\partial\mathscr l}{\partial u_\pm}=\rho_\pm d\frac{\partial\mathscr l }{\partial \rho_\pm}+\varsigma_\pm d \frac{\partial \mathscr l }{\partial s_\pm}+(s_\pm-\varsigma_\pm)d\frac{\partial\mathscr l}{\partial s_\pm},\\
\delta \rho_\pm:&\;\;\;D_t w_\pm=-\frac{\partial\mathscr l}{\partial \rho_\pm},\\
\delta w_\pm:&\;\;\;\bar D_t \rho_\pm=0,\\
\delta s_\pm:&\;\;\;D_t \gamma_\pm=-\frac{\partial\mathscr l}{\partial s_\pm},\\
\delta \gamma_\pm:&\;\;\;\bar{D}_t(s_\pm-\varsigma_\pm)=0.
\end{align*}
The form of the RHS of the first equation arises from the following equalities
\[
\partial_t(\rho\nabla w)+\mathscr L_u (\rho\nabla w)=\bar{D_t}\rho\nabla w+\rho \nabla D_t w
=-\rho d\frac{\partial\mathscr l}{\partial \rho},
\]
\[
\partial_t((s-\varsigma)\nabla \gamma)+\mathscr L_u ((s-\varsigma)\nabla \gamma)=\bar{D_t}(s-\varsigma)\nabla \gamma+(s-\varsigma) \nabla D_t \gamma
=-(s-\varsigma)d\frac{\partial\mathscr l}{\partial s}.
\]
We hence get Equations \eqref{genEulerLag_rho_s}.
Recall that we regard $\rho$ as a density rather than as a scalar function, and therefore $\mathscr L_u (\rho\,\nabla w)$ is the Lie derivative of a $1$-form density.

\medskip

Let us split the contribution of $\partial\Omega_\pm$ as follows:
\[
\int_{\partial \Omega_\pm}=\int_{t_0}^{t_1}\int_{S_\pm}+\int_\Sigma.
\]

\noindent\textit{Boundary conditions.} We can get the boundary conditions on $S_\pm$ from the following equation 
$$
\begin{aligned}
0=&\int_{t_0}^{t_1}{\rm d}t \int_{\Sigma_\pm}\Bigg(\frac{\partial\mathscr l}{\partial u_\pm}\cdot \xi_\pm n_t^\pm + \frac{\partial\mathscr l}{\partial u_\pm}(\xi_\pm) u_\pm \cdot n_x^{\pm}+\left(\mathscr l_\pm-\rho_\pm\frac{\partial\mathscr l}{\partial \rho_\pm}-s_\pm\frac{\partial\mathscr l}{\partial s_\pm}\right) \xi_\pm\cdot n_x^\pm\\
+&\rho_\pm\xi_\pm\cdot\nabla w_\pm n_t^\pm+\rho_\pm\xi_\pm\cdot\nabla w_\pm u_\pm\cdot n_x^\pm+(s_\pm-\varsigma_\pm)\xi_\pm\cdot\nabla \gamma_\pm n_t^\pm+(s_\pm-\varsigma_\pm)\xi_\pm\cdot\nabla \gamma_\pm u_\pm \cdot n_x^\pm\Bigg)\,{\rm d}a.
\end{aligned}
$$
For fluids in a fixed domain, along $S_\pm$ one has 
\[
n_t^\pm = -\,u_\pm\cdot n_x^\pm = 0,
\qquad 
\xi_\pm\cdot n_x^\pm = 0 .
\]
Hence the boundary is stationary and the admissible variations are tangent to it. 
Under these conditions, the above equation is automatically satisfied.

For free-boundary fluids, along $S_\pm$ we have 
$$
\rho_\pm\xi_\pm\cdot\nabla w_\pm n_t^\pm+\rho_\pm\xi_\pm\cdot\nabla w_\pm u_\pm\cdot n_x^\pm=0,
$$
$$
(s_\pm-\varsigma_\pm)\xi_\pm\cdot\nabla \gamma_\pm n_t^\pm+(s_\pm-\varsigma_\pm)\xi_\pm\cdot\nabla \gamma_\pm u_\pm \cdot n_x^\pm=0,
$$
and
$$
\frac{\partial\mathscr l}{\partial u_\pm}\cdot \xi_\pm n_t^\pm + \frac{\partial\mathscr l}{\partial u_\pm}(\xi_\pm) u_\pm \cdot n_x^{\pm}=0.
$$
These identities follow from the kinematic boundary condition for a free surface,
\[
n_t^\pm = -\,u_\pm\!\cdot\! n_x^\pm \quad \text{on } S_\pm.
\]
Consequently, the natural boundary condition along $S_\pm$ reduces to
Equation~\eqref{genBounCon_rho_s}.

\medskip

\noindent\textit{Rankine–Hugoniot condition for mass.}
Assuming that  $\delta w_\pm$ is continuous across the shock surface $\Gamma$, i.e., $\delta w_+=\delta w_-=\delta w$, the boundary terms on $\Gamma$ imply 
$$
\delta w: \;\;\; \rho_+ n_t^++\rho_+ u_+\cdot n_x^+=\rho_- n_t^-+\rho_- u_-\cdot n_x^-.
$$
Using $n_t^+=-v_s=-n_t^-$ and $n_x^+= n=-n_x^-$, we obtain the Rankine--Hugoniot jump condition for mass \eqref{eq:RH_mass_rho_s}.

\medskip

\noindent\textit{Rankine–Hugoniot condition for entropy.} We assume $\delta \gamma_\pm$ are continuous across the shock surface $\Gamma$, i.e., $\delta\gamma_+=\delta\gamma_-=\delta \gamma$, then we have 
$$
\delta \gamma: \;\;\; (s_+-\varsigma_+) n_t^++(s_+-\varsigma_+) u_+\cdot n_x^+=(s_--\varsigma_-) n_t^-+(s_--\varsigma_-) u_-\cdot n_x^-,
$$
and on the shock surface $\Gamma$, we have $n_t^+=-v_s=-n_t^-$ and $n_x^+= n=-n_x^-$, which gives us the Rankine--Hugoniot jump condition for entropy \eqref{RH_entropy_rho_s}.

\medskip

\noindent\textit{Rankine--Hugoniot equations for momentum.} 
Finally, we derive the conditions on the singular surface $\Gamma$. For any $\xi_\pm$ on $\Gamma$, we have 
$$
\begin{aligned}
0=&\frac{\partial\mathscr l}{\partial u_\pm}\cdot \xi_\pm n_t^\pm + \frac{\partial\mathscr l}{\partial u_\pm}(\xi_\pm) u_\pm \cdot n_x^{\pm}-\varsigma_\pm \frac{\partial\mathscr l}{\partial s_\pm} \xi_\pm\cdot n_x^\pm +\rho_\pm\xi_\pm\cdot\nabla w_\pm n_t^\pm+\rho_\pm\xi_\pm\cdot\nabla w_\pm u_\pm\cdot n_x^\pm\\
&+(s_\pm-\varsigma_\pm)\xi_\pm\cdot\nabla \gamma_\pm n_t^\pm+(s_\pm-\varsigma_\pm)\xi_\pm\cdot\nabla \gamma_\pm u_\pm\cdot n_x^\pm.
\end{aligned}
$$

Using the geometric identities
\[
n_t=-v_s,\qquad n_x= n,
\]
on $\Gamma$, we get
$$
\begin{aligned}
v_s\frac{\partial\mathscr l_\pm}{\partial u_\pm}(\cdot)&=\left(\frac{\partial\mathscr l_\pm}{\partial u_\pm}\otimes u_\pm-\varsigma_\pm\frac{\partial\mathscr l_\pm}{\partial s_\pm}\delta\right)\cdot n-v_s\rho_\pm\nabla w_\pm(\cdot)+\rho_\pm\nabla w_\pm(\cdot) u_\pm\cdot n\\
&-v_s(s_\pm-\varsigma_\pm)\nabla \gamma_\pm(\cdot) +(s_\pm-\varsigma_\pm)\nabla \gamma_\pm(\cdot) u_\pm\cdot n.   
\end{aligned}
$$
Then, for any $\xi^s=\xi_s\cdot n$ on $\Gamma$, we get $\mathscr l_++\rho_+ D_t w_++(s_+ - \varsigma_+)D_t\gamma_+=\mathscr l_-+\rho_-D_t w_-+(s_- - \varsigma_-)D_t\gamma_-$. Together with the RH condition for entropy, we obtain the Rankine--Hugoniot condition on $\Gamma$ given by Equation \eqref{genLagRH_rho_s}.
\end{proof}

\paragraph{Bulk, interface, and boundary equations in standard PDE notation.}
Theorem~\ref{RHvar_rho_s} furnishes a variational formulation of the compressible Euler equations with entropy that is directly comparable to the classical conservation--law framework.

The first two equations in \eqref{genEulerLag_rho_s} are transport equations for mass density and entropy density,
\[
\partial_t \rho_\pm + \mathscr L_{u_\pm} \rho_\pm = 0,
\qquad
\partial_t s_\pm + \mathscr L_{u_\pm} s_\pm = 0,
\]
which, in divergence form, read
\[
\partial_t \rho_\pm + \nabla\cdot(\rho_\pm u_\pm) = 0,
\qquad
\partial_t s_\pm + \nabla\cdot(s_\pm u_\pm) = 0.
\]
The third equation in \eqref{genEulerLag_rho_s} expresses momentum balance. Defining
\(
m_\pm := \frac{\partial \mathscr l}{\partial u_\pm},
\)
and using $\mathscr L_u m = \nabla_u m + \nabla u^T m + m \operatorname{div}u =  \nabla\cdot (m \otimes u)+\nabla u^T m$,
it can be written in the form
\[
\partial_t m_\pm
+ \nabla\cdot (m_\pm \otimes u_\pm)+\nabla u_\pm^T m_\pm
- \rho_\pm \nabla \frac{\partial \mathscr l}{\partial \rho_\pm}-s_\pm \nabla \frac{\partial \mathscr l}{\partial s_\pm} = 0.
\]
For the Lagrangian density $\mathscr l =\frac 12\rho|u|^2-\varepsilon(\rho,s)$, we get
$\nabla u_\pm^T m_\pm
- \rho_\pm \nabla \frac{\partial \mathscr l}{\partial \rho_\pm}-s_\pm \nabla \frac{\partial \mathscr l}{\partial s_\pm}=\rho_\pm \nabla \frac{\partial \varepsilon}{\partial \rho_\pm}+s_\pm \nabla \frac{\partial \varepsilon}{\partial s_\pm}=\nabla p_\pm$, where
\begin{equation}\label{pressure_def_remark}
p=\mathscr l -\rho\frac{\partial\mathscr l }{\partial\rho}-s \frac{\partial \mathscr l}{\partial s}=-\varepsilon+\rho\frac{\partial \varepsilon}{\partial \rho}+s \frac{\partial \varepsilon}{\partial s}
\end{equation}
is the pressure. Hence, for standard choices of $\mathscr l $, this reduces to the familiar Euler form
\[
\partial_t(\rho u)
+\nabla\cdot(\rho u\otimes u)
+\nabla p
=0 .
\]

The boundary condition \eqref{genBounCon_rho_s} then becomes the pressure condition $p = 0$ on $S_\pm$ for free-boundary fluids.

The jump condition \eqref{genLagRH_rho_s} expresses momentum conservation across the moving interface $\Gamma$. Using the pressure definition \eqref{pressure_def_remark}, it takes the classical Rankine--Hugoniot form
\[
v_s [\![\rho u]\!]
=
[\![\rho u \otimes u]\!]  n
+ [\![p]\!]  n,
\]
where $v_s$ denotes the normal speed of $\Gamma$ and $ n$ its unit normal vector.

Similarly, \eqref{eq:RH_mass_rho_s} is the standard Rankine--Hugoniot condition for mass conservation,
\[
v_s [\![\rho]\!] = [\![\rho u]\!] \cdot  n.
\]

Altogether, Theorem \ref{RHvar_rho_s} shows that a single variational framework - viewed as a consistent extension of Hamilton's principle to irreversible processes - recovers the classical compressible Euler equations in the bulk, along with the Rankine–Hugoniot jump conditions for mass and momentum, as well as the boundary conditions. The Rankine–Hugoniot condition for energy is derived from the variational principole in the following section.

\subsection{Time derivative of the energy balance}
\begin{lemma}\label{lem:RH_E_rho_s}
The critical action principle \eqref{VP_thermo1}-\eqref{constraint_thermo1} yields the Rankine--Hugoniot condition for energy:
\begin{equation}\label{RH_E_rho_s}
v_s\left[\!\left[E\right]\!\right]
-\left[\!\left[(E+p)u\right]\!\right]\cdot n=0.
\end{equation}
\end{lemma}
\begin{proof}
Recall that the energy density associated with the Lagrangian is
\(
E=\frac{\partial \mathscr l}{\partial u}\cdot u-\mathscr l,
\)
and that the pressure is given by the thermodynamic relation
\(
p=\mathscr l-\rho\frac{\partial \mathscr l}{\partial \rho}-s\frac{\partial \mathscr l}{\partial s}.
\)
Substituting the expressions for $E$ and $p$, we get
\begin{equation}\label{RH_E_integrand}
v_s\left[\!\left[E\right]\!\right]
-\left[\!\left[(E+p)u\right]\!\right]\cdot n=v_s\left[\!\left[\frac{\partial \mathscr l}{\partial u}\cdot u-\mathscr l\right]\!\right]
-\left[\!\left[\left(\frac{\partial \mathscr l}{\partial u}\cdot u
-\rho\frac{\partial \mathscr l}{\partial \rho}
-s\frac{\partial \mathscr l}{\partial s}\right)u\right]\!\right]\cdot n.
\end{equation}
From the proof of Theorem \ref{RHvar_rho_s}, we know that on $\Gamma$, one has 
$$
\begin{aligned}
v_s\frac{\partial\mathscr l_\pm}{\partial u_\pm}(\cdot)&=\left(\frac{\partial\mathscr l_\pm}{\partial u_\pm}\otimes u_\pm-\varsigma_\pm\frac{\partial\mathscr l_\pm}{\partial s_\pm}\delta\right)\cdot n-v_s\rho_\pm\nabla w_\pm(\cdot)+\rho_\pm\nabla w_\pm(\cdot) u_\pm\cdot n\\
&\quad -v_s(s_\pm-\varsigma_\pm)\nabla \gamma_\pm(\cdot) +(s_\pm-\varsigma_\pm)\nabla \gamma_\pm(\cdot) u_\pm\cdot n.   
\end{aligned}
$$
Evaluating this identity on $u_\pm$ yields
$$
\begin{aligned}
v_s\frac{\partial\mathscr l_\pm}{\partial  u_\pm}\cdot u_\pm&=\left(\frac{\partial\mathscr l_\pm}{\partial  u_\pm}\cdot  u_\pm-\varsigma_\pm\frac{\partial\mathscr l_\pm}{\partial s_\pm}\right) u_\pm\cdot  n-v_s\rho_\pm\nabla w_\pm\cdot u_\pm+\rho_\pm\nabla w_\pm\cdot u_\pm u_\pm \cdot n\\
&\quad -v_s(s_\pm-\varsigma_\pm)\nabla \gamma_\pm\cdot u_\pm +(s_\pm-\varsigma_\pm)\nabla \gamma_\pm\cdot u_\pm u_\pm \cdot n.  
\end{aligned}
$$
Taking the jump across $\Gamma$ give
$$
\begin{aligned}
v_s\left[\!\left[\frac{\partial \mathscr l}{\partial  u}\cdot u\right]\!\right]&=\left[\!\left[\left(\frac{\partial \mathscr l}{\partial  u}\cdot u-\varsigma\frac{\partial \mathscr l}{\partial s}\right) u\right]\!\right]\cdot n-\left[\!\left[\rho  u\cdot\nabla w\right]\!\right]v_s+\left[\!\left[\rho  u\cdot\nabla w  u\right]\!\right]\cdot  n\\
&\quad -\left[\!\left[(s-\varsigma)  u\cdot\nabla \gamma\right]\!\right]v_s+\left[\!\left[(s-\varsigma)  u\cdot\nabla \gamma  u\right]\!\right]\cdot  n.
\end{aligned}
$$
Substituting this expression into Equation \eqref{RH_E_integrand} yields
$$
\begin{aligned}
&v_s\left[\!\left[\frac{\partial \mathscr l}{\partial  u}\cdot u-\mathscr l\right]\!\right]-\left[\!\left[\left(\frac{\partial \mathscr l}{\partial  u}\cdot u- \rho\frac{\partial \mathscr l}{\partial \rho}-s\frac{\partial \mathscr l}{\partial s}\right) u\right]\!\right]\cdot n\\
&=v_s\left[\!\left[-\mathscr l\right]\!\right]+\left[\!\left[\rho\frac{\partial \mathscr l}{\partial \rho} u\right]\!\right]\cdot n+\left[\!\left[s\frac{\partial \mathscr l}{\partial s} u\right]\!\right]\cdot n-\left[\!\left[\varsigma\frac{\partial \mathscr l}{\partial s} u\right]\!\right]\cdot n-\left[\!\left[\rho  u\cdot\nabla w\right]\!\right]v_s+\left[\!\left[\rho  u\cdot\nabla w  u\right]\!\right]\cdot  n\\
&\quad -\left[\!\left[(s-\varsigma)  u\cdot\nabla \gamma\right]\!\right]v_s+\left[\!\left[(s-\varsigma)  u\cdot\nabla \gamma  u\right]\!\right]\cdot  n.
\end{aligned}
$$
On the shock surface $\Gamma$, we have
\[
\left[\!\left[-\mathscr l\right]\!\right]
=
\left[\!\left[\rho D_t w + (s-\varsigma)D_t\gamma\right]\!\right].
\]
Moreover, $\partial_t w$, $\nabla_x w$, $\partial_t\gamma$, and $\nabla_x\gamma$ are continuous across $\Gamma$. Using the Rankine--Hugoniot conditions for mass and entropy, we obtain
$$
\begin{aligned}
&v_s\left[\!\left[\rho D_t w\right]\!\right]+\left[\!\left[\rho\frac{\partial \mathscr l}{\partial \rho} u\right]\!\right]\cdot n-\left[\!\left[\rho  u\cdot\nabla w\right]\!\right]v_s+\left[\!\left[\rho  u\cdot\nabla w  u\right]\!\right]\cdot  n\\
&=v_s\left[\!\left[\rho \partial_t w\right]\!\right]+\left[\!\left[\rho\frac{\partial \mathscr l}{\partial \rho} u\right]\!\right]\cdot n+\left[\!\left[\rho  u\cdot\nabla w  u\right]\!\right]\cdot  n\\
&=\partial_t w\left[\!\left[\rho  u\right]\!\right]\cdot  n+\left[\!\left[\rho\frac{\partial \mathscr l}{\partial \rho} u\right]\!\right]\cdot n+\left[\!\left[\rho  u\cdot\nabla w  u\right]\!\right]\cdot  n\\
&=\left[\!\left[\rho D_t w u\right]\!\right]\cdot  n+\left[\!\left[\rho\frac{\partial \mathscr l}{\partial \rho} u\right]\!\right]\cdot n=0,
\end{aligned}
$$
and similarly, 
$$
\begin{aligned}
v_s\left[\!\left[(s-\varsigma) D_t \gamma\right]\!\right]+\left[\!\left[(s-\varsigma)\frac{\partial \mathscr l}{\partial s} u\right]\!\right]\cdot n-\left[\!\left[(s-\varsigma)  u\cdot\nabla \gamma\right]\!\right]v_s+\left[\!\left[(s-\varsigma)  u\cdot\nabla \gamma  u\right]\!\right]\cdot  n=0.
\end{aligned}
$$
All contributions therefore cancel, and we get the Rankine--Hugoniot condition \eqref{RH_E_rho_s}.
\end{proof}

From this lemma, the following result follows immediately.

\begin{theorem} The Euler–Lagrange equations associated with the variational principle \eqref{VP_thermo1}-\eqref{constraint_thermo1} imply conservation of energy, i.e.,
\begin{equation}\label{time_change_E_rho_s}
\frac{d}{dt}\mathscr E=0.
\end{equation}
\end{theorem}
\begin{proof}
On the free boundary $\partial M$, the natural boundary condition implies
\[
p=\mathscr l-\rho\frac{\partial \mathscr l}{\partial \rho}-s\frac{\partial \mathscr l}{\partial s}=0.
\]
Using Proposition~\ref{prop:energy}, also valid for the full compressible case, the time derivative of the total energy reduces to an integral over the shock surface $\Gamma_t$,
\[
\frac{d}{dt}\mathscr E
=\int_{\Gamma_t}
v_s\left[\!\left[E\right]\!\right]
-\left[\!\left[(E+p)u\right]\!\right]\cdot n \,{\rm d}a.
\]
Hence, by Lemma~\ref{lem:RH_E_rho_s}, we get \eqref{time_change_E_rho_s}.
\end{proof}

\section{Variational formulation for shocks with entropy production and heat conduction}\label{sec:entropy}

In this section, we show that our variational framework can incorporate irreversible processes in fluids by focusing on heat conduction. This setting highlights the essential distinction between the variables $s$ and $\varsigma$, which in this case no longer coincide.

Following the variational approach to nonequilibrium thermodynamics, the constraints \eqref{constraint_thermo1} must be modified to account for irreversible processes, in the sense of the general framework \eqref{constraints} involving phenomenological and variational constraints as developed in \cite{GBYo2019}. We denote by $j_s^\pm$ the entropy flux, which may include contributions from heat conduction and other irreversible mechanisms.

The critical action principle then takes the form
\begin{equation}\label{genLagrangian_thermo}
\begin{aligned}
\delta \int_{t_0}^{t_1} 
\Bigg[
&\int_{M_-(t)}
\left(
\mathscr l (u_-,\rho_-,s_-)
+\rho_-D_t w_-
+(s_- - \varsigma_-)D_t\gamma_-
\right) \,{\rm d}x \\
&+
\int_{M_+(t)}
\left(
\mathscr l (u_+,\rho_+,s_+)
+\rho_+D_t w_+
+(s_+ - \varsigma_+)D_t\gamma_+
\right) \,{\rm d}x
\Bigg] \,{\rm d}t
=0 .
\end{aligned}
\end{equation}
subject to the phenomenological and variational constraints
\begin{equation}\label{Constraints_js}
\frac{\partial \mathscr l_\pm}{\partial s_\pm}\,
\bar{D}_t \varsigma _\pm
=
j_s^\pm \cdot \nabla D_t \gamma _\pm
\qquad \text{and}\qquad 
\frac{\partial \mathscr l_\pm}{\partial s_\pm}\,
\bar{D}_\delta \varsigma _\pm
=
j_s^\pm \cdot \nabla D_\delta \gamma _\pm.
\end{equation}
The critical condition must hold for arbitrary variations $\delta\Gamma$, $\delta\rho_\pm$, $\delta w_\pm$, and $\delta\gamma_\pm$, together with velocity variations of the form
\[
\delta u_\pm = \partial_t \xi_\pm + \mathcal{L}_{u_\pm}\xi_\pm.
\]
We refer to Remark \ref{remark_dAlembert} for the D'Alembert structure of the principle.

\subsection{Variational derivation of the compressible Euler system with entropy production}\label{sec:compr_Euler_entropy}

The following result extends Theorem \ref{RHvar_rho_s} to the case including heat conduction. The presence of entropy flux modifies the entropy equation and the Rankine--Hugoniot relation for energy and yields a new boundary condition.

\begin{theorem}\label{RHvar_thermo} The critical action principle \eqref{genLagrangian_thermo}--\eqref{Constraints_js} implies:
\begin{itemize}
\item[(1)] \textbf{Bulk equations}.
\begin{equation}\label{genEulerLag_thermo}
\left\{
\begin{aligned}
&\partial_t \rho_\pm+\mathscr L_{u_\pm}\rho_\pm=0,\\
&\frac{\partial \mathscr l}{\partial s_\pm}\left(\partial_t s_\pm+\mathscr L_{u_\pm}s_\pm\right)
=-{\rm div}\!\left( j_s^\pm\frac{\partial \mathscr l}{\partial s_\pm}\right),\\
&\partial_t \frac{\partial\mathscr l}{\partial u_\pm}
+\mathscr L_{u_\pm}\frac{\partial\mathscr l}{\partial u_\pm}
=\rho_\pm d\frac{\partial\mathscr l}{\partial \rho_\pm}
+s_\pm d\frac{\partial\mathscr l}{\partial s_\pm}.
\end{aligned}
\right.
\end{equation}
\item[(2)] \textbf{Boundary conditions}. As in Theorem \ref{RHvar_rho_s}, together with the insulated condition 
\[
j_s^\pm\cdot n=0,
\]
for both fixed and free boundaries.
\item[(3)] \textbf{Jump conditions across interfaces}. The Rankine-Hugoniot conditions for momentum and mass \eqref{genLagRH_rho_s} and \eqref{eq:RH_mass_rho_s}, together with the modified entropy jump condition
\begin{equation}\label{RH_entropy}
v_s\left[\!\left[s-\varsigma\right]\!\right]
=
\left[\!\left[(s-\varsigma)u+j_s\right]\!\right]\cdot n.
\end{equation}
\end{itemize}
\end{theorem}

\begin{proof}
For arbitrary variations $q_\epsilon$ of the curve $q(t)\in \mathcal{Q}_{\rm free}$ (or $\mathcal{Q}_{\rm fix}$)
with fixed endpoints, the critical action principle gives
\begin{equation}\label{stan_eq_thermo}
\begin{aligned}
0=&\delta\int_{t_0}^{t_1}{\rm d}t
\left[
    \int_{M_-(t)}
        \left( \mathscr l (u_-,\rho_-,s_-)
        + \rho_- D_t w_-
        + (s_- - \varsigma_-) D_t \gamma_- \right)\,{\rm d}x
%\right. \\
%&\left.
   + \int_{M_+(t)}
        \left( \mathscr l (u_+,\rho_+,s_+)
        + \rho_+ D_t w_+
        + (s_+ - \varsigma_+) D_t \gamma_+ \right)\,{\rm d}x
\right]\\
=&\int_{t_0}^{t_1}{\rm d}t\int_{M_+}\left(\frac{\partial\mathscr l}{\partial u_+}\cdot\delta u_+ +\frac{\partial\mathscr l}{\partial \rho_+}\cdot\delta \rho_+ +\frac{\partial\mathscr l}{\partial s_+}\cdot\delta s_++\delta\rho_+D_t w_++\rho_+\delta D_tw_++\delta(s_+ - \varsigma_+)D_t\gamma_++(s_+ - \varsigma_+)\delta D_t\gamma_+ \right){\rm d}x\\
+&\int_{t_0}^{t_1}{\rm d}t\left(\int_{S_+}(\mathscr l (u_+,\rho_+,s_+)+ \rho_+ D_t w_++ (s_+ - \varsigma_+) D_t \gamma_+)\xi_+\cdot n_+ {\rm d}a +\int_{\Gamma}(\mathscr l (u_+,\rho_+,s_+)+ \rho_+ D_t w_++ (s_+ - \varsigma_+) D_t \gamma_+)\xi_s\cdot n_+ {\rm d}a\right)\\
+&\int_{t_0}^{t_1}{\rm d}t\int_{M_-}\left(\frac{\partial\mathscr l}{\partial u_-}\cdot\delta u_- +\frac{\partial\mathscr l}{\partial \rho_-}\cdot\delta \rho_- +\frac{\partial\mathscr l}{\partial s_-}\cdot\delta s_-+\delta\rho_-D_t w_-+\rho_-\delta D_tw_-+\delta(s_- - \varsigma_-)D_t\gamma_-+(s_- - \varsigma_-)\delta D_t\gamma_- \right){\rm d}x\\
+&\int_{t_0}^{t_1}{\rm d}t\left(\int_{S_-}(\mathscr l (u_-,\rho_-,s_-)+ \rho_- D_t w_-+ (s_- - \varsigma_-) D_t \gamma_-)\xi_-\cdot n_- {\rm d}a +\int_{\Gamma}(\mathscr l (u_-,\rho_-,s_-)+ \rho_- D_t w_-+ (s_- - \varsigma_-) D_t \gamma_-)\xi_s\cdot n_- {\rm d}a\right).\\
\end{aligned}
\end{equation}
Note that if the domain is fixed, then on $S_\pm$, $\xi_\pm\cdot n_\pm=0$.

We know $\delta u_\pm=\partial_t\xi_\pm+[u_\pm,\xi_\pm],$
where $\xi_\pm = \delta\varphi_\pm \circ \varphi_\pm^{-1}$
is an arbitrary time dependent vector field vanishing
at $t = t_0,\, t_1$.

\medskip

\noindent\textit{Bulk equations.} Using
\[
\delta D_t\gamma
=
D_t\delta\gamma
+
(\partial_t\xi+\mathscr L_u\xi,\nabla)\gamma,
\]
\[
\rho D_t\delta w=\bar D_t(\rho\delta w)-\bar D_t\rho\,\delta w,
\]
\[
\partial_t(\rho\nabla w)+\mathscr L_u (\rho\nabla w)=\bar{D_t}\rho\nabla w+\rho \nabla D_t w
=-\rho \nabla\frac{\partial\mathscr l}{\partial \rho},
\]
\[
\partial_t((s-\varsigma)\nabla \gamma)+\mathscr L_u ((s-\varsigma)\nabla \gamma)=\bar{D_t}(s-\varsigma)\nabla \gamma+(s-\varsigma) \nabla D_t \gamma
= -{\rm div}j_s\nabla\gamma-(s-\varsigma)\nabla\frac{\partial\mathscr l}{\partial s},
\]
together with the variational constraint
\[
\begin{aligned} &-D_t\gamma\delta\varsigma=\frac{\partial \mathscr l }{\partial s}\delta \varsigma=-\frac{\partial \mathscr l }{\partial s}{\rm div}(\varsigma\xi)+ j_s\cdot\nabla(\delta\gamma+(\xi,\nabla)\gamma)\\ =&\varsigma d\frac{\partial \mathscr l }{\partial s}\cdot\xi-{\rm div}\left(\varsigma \frac{\partial \mathscr l }{\partial s}\xi\right)-{\rm div} j_s\delta\gamma+{\rm div}\left( j_s\delta\gamma\right)-{\rm div} j_s \xi\cdot\nabla\gamma+{\rm div}\left( j_s \xi\cdot\nabla\gamma\right), 
\end{aligned}
\]
we integrate by parts in spacetime, collect the independent variations, and obtain in $\Omega_\pm$:
\begin{align*}
\xi_\pm:&\quad
\partial_t\frac{\partial\mathscr l}{\partial u_\pm}
+\mathscr L_{u_\pm}\frac{\partial\mathscr l}{\partial u_\pm}
=
-\rho_\pm\nabla D_t w_\pm
-\varsigma_\pm\nabla D_t\gamma_\pm
-{\rm div}j_s^\pm\nabla\gamma_\pm
+{\rm div}j_s^\pm\nabla\gamma_\pm+(s_\pm-\varsigma_\pm)\nabla\frac{\partial\mathscr l}{\partial s_\pm},\\
\delta\rho_\pm:&\quad
D_t w_\pm=-\frac{\partial\mathscr l}{\partial\rho_\pm},\\
\delta w_\pm:&\quad
\bar D_t\rho_\pm=0,\\
\delta s_\pm:&\quad
D_t\gamma_\pm=-\frac{\partial\mathscr l}{\partial s_\pm},\\
\delta\gamma_\pm:&\quad
\bar D_t(s_\pm-\varsigma_\pm)+{\rm div}j_s^\pm=0.
\end{align*}

The momentum equation reduces to
\[
\partial_t\frac{\partial\mathscr l}{\partial u_\pm}
+\mathscr L_{u_\pm}\frac{\partial\mathscr l}{\partial u_\pm}
=
\rho_\pm d\frac{\partial\mathscr l}{\partial\rho_\pm}
+
s_\pm d\frac{\partial\mathscr l}{\partial s_\pm},
\]
which gives \eqref{genEulerLag_thermo}.

\medskip

\noindent\textit{Boundary conditions on $S_\pm$.}
For both fixed- and free-boundary cases, we obtain
$$
j_s^\pm\cdot n_x^\pm\delta\gamma=0,\qquad \text{on }S_\pm.
$$
Since $\delta\gamma$ is is arbitrary, it follows that the insulated boundary condition
$j_s^\pm\cdot n=0$ holds.

For fluids in a fixed domain, along $S_\pm$ one has 
\[
n_t^\pm = -\,u_\pm\cdot n_x^\pm = 0,
\qquad 
\xi_\pm\cdot n_x^\pm = 0 .
\]
Hence the boundary is stationary and the admissible variations are tangent to it. 
Under these conditions, no additional boundary conditions on $p_\pm=\mathscr l_\pm
-\rho_\pm\frac{\partial\mathscr l}{\partial\rho_\pm}
-s_\pm\frac{\partial\mathscr l}{\partial s_\pm}$ arise.
%the above equation is automatically satisfied.

For free-boundary fluids, along $S_\pm$, since $n_t^\pm=-u_\pm\!\cdot\! n_x^\pm$, we have
\[
\rho_\pm\xi_\pm\!\cdot\!\nabla w_\pm n_t^\pm
+\rho_\pm\xi_\pm\!\cdot\!\nabla w_\pm u_\pm\cdot n_x^\pm=0,
\]
\[
(s_\pm-\varsigma_\pm)\xi_\pm\!\cdot\!\nabla\gamma_\pm n_t^\pm
+(s_\pm-\varsigma_\pm)\xi_\pm\!\cdot\!\nabla\gamma_\pm u_\pm\cdot n_x^\pm=0,
\]
and
\[
\frac{\partial\mathscr l}{\partial u_\pm}\!\cdot\!\xi_\pm n_t^\pm
+
\frac{\partial\mathscr l}{\partial u_\pm}(\xi_\pm) u_\pm\cdot
n_x^\pm
=0.
\]

The remaining terms reduce to
\[
\left(
\mathscr l_\pm
-\rho_\pm\frac{\partial\mathscr l}{\partial\rho_\pm}
-s_\pm\frac{\partial\mathscr l}{\partial s_\pm}
\right)\xi_\pm\cdot n_x^\pm
+
j_s^\pm\cdot n_x^\pm\nabla\gamma_\pm\!\cdot\!\xi_\pm=0,
\]
which gives the boundary condition \eqref{genBounCon_rho_s} since $j_s^\pm\cdot n_x^\pm=0$ on $S_\pm$.

\medskip

\noindent\textit{Rankine–Hugoniot conditions.} Assume $\delta w_+=\delta w_-=\delta w$ on $\Gamma$.  
Boundary terms yield
\[
\rho_+ n_t^+ + \rho_+ u_+\cdot n_x^+
=
\rho_- n_t^- + \rho_- u_-\cdot n_x^-.
\]
Since on $\Gamma$,
\[
n_t^+=-v_s=-n_t^-,
\qquad
n_x^+=n=-n_x^-,
\]
we obtain the mass Rankine–Hugoniot condition \eqref{eq:RH_mass_rho_s}.

Similarly, continuity of $\delta\gamma$ gives
\[
(s_+-\varsigma_+)n_t^+
+(s_+-\varsigma_+)u_+\cdot n_x^+
+j_s^+\cdot n_x^+
=
(s_--\varsigma_-)n_t^-
+(s_--\varsigma_-)u_-\cdot n_x^-
+j_s^-\cdot n_x^-,
\]
which yields the entropy jump condition \eqref{RH_entropy}.

For variations $\xi_\pm$ on $\Gamma$, using
\[
n_t^+=-v_s,\qquad n_x^+=n,
\]
we obtain
$$
\begin{aligned}
v_s\frac{\partial\mathscr l_\pm}{\partial u_\pm}(\cdot)=&\left(\frac{\partial\mathscr l_\pm}{\partial u_\pm}\otimes u_\pm-\varsigma_\pm\frac{\partial\mathscr l_\pm}{\partial s_\pm}\delta\right)\cdot n-v_s\rho_\pm\nabla w_\pm(\cdot)+\rho_\pm\nabla w_\pm(\cdot) u_\pm \cdot n\\
&+j_s^\pm \cdot n\nabla \gamma_\pm(\cdot)-v_s(s_\pm-\varsigma_\pm)\nabla \gamma_\pm(\cdot) +(s_\pm-\varsigma_\pm)\nabla \gamma_\pm(\cdot) u_\pm \cdot n.   
\end{aligned}
$$
Then, for any $\xi^s=\xi_s\cdot n$ on $\Gamma$, we get $\mathscr l_++\rho_+D_tw_++(s_+ - \varsigma_+)D_t\gamma_+=\mathscr l_-+\rho_-D_tw_-+(s_- - \varsigma_-)D_t\gamma_-$. 
Together with the RH conditions for mass and entropy,
these extra terms cancel, giving the momentum RH condition
\eqref{genLagRH_rho_s}.
\end{proof}

\subsection{Global energy balance and conservation across a shock}

In the smooth regions, the total energy balance for a fluid with Lagrangian $\mathscr l(u,\rho,s)$ takes the general form 
\begin{equation}\label{eq:baro_energy_thermo}
\partial_t E+\nabla \cdot\left((E+p) u\right)
=-\nabla\cdot\left( j_s T\right),
\end{equation}
where 
\[
E=\frac{\partial \mathscr l}{\partial u}\cdot u - \mathscr l, \qquad p=\mathscr l-\rho\frac{\partial\mathscr l}{\partial\rho}
-s\frac{\partial\mathscr l}{\partial s},
\qquad 
T=-\frac{\partial \mathscr l}{\partial s}.
\]
As earlier, for each side of the shock, we define the total energy by
\[
\mathscr E_\pm(t)=\int_{M_\pm(t)}E_\pm\,dx,
\qquad
\mathscr E(t)=\mathscr E_+(t)+\mathscr E_-(t),
\]
and the extension of Proposition \ref{prop:energy} to the presence of entropy flux is as follows:

\begin{proposition}\label{prop:energy_thermo} The time derivative of the total energy satisfies
\begin{equation}\label{time_change_E_thermo}
\frac{d}{dt}\mathscr E=\int_{\Gamma_t} v_s\left[\!\left[E\right]\!\right]-\left[\!\left[(E+p) u+T j_s \right]\!\right]\cdot n \,{\rm d}a-\int_{\partial M} p u\cdot n \,{\rm d}a.
%-\int_{\partial M} T j_s\cdot  n.
\end{equation}
\end{proposition}
\begin{proof}
Applying Reynolds' transport theorem on each moving domain $M_\pm(t)$ and using \eqref{eq:baro_energy_thermo}, we obtain
\[
\frac{d}{dt}\mathscr E_\pm
=\int_{\partial M_\pm}
E_\pm u^s\cdot n_\pm \,{\rm d}a
-\int_{\partial M_\pm}(E_\pm+p_\pm)u_\pm\cdot n_\pm \,{\rm d}a
-\int_{\partial M_\pm}T_\pm j_s^\pm\cdot n_\pm \,{\rm d}a .
\]
On $S_\pm$ we have $u^s\cdot n_\pm=u_\pm\cdot n_\pm$, while on $\Gamma$ we use $n=n_+=-n_-$.  
Summing the two sides and using $j_s\cdot n=0$ on $\partial M$ yields \eqref{time_change_E_thermo}.
\end{proof}

\begin{lemma}\label{lem:RH_E_lambda_thermo}
The critical action principle \eqref{genLagrangian_thermo}--\eqref{Constraints_js} yields the Rankine--Hugoniot condition for energy:
\begin{equation}\label{RH_E_lambda_thermo}
v_s\left[\!\left[E\right]\!\right]-\left[\!\left[(E+p) u+T j_s\right]\!\right]\cdot n=0.
\end{equation}
\end{lemma}
\begin{proof}
Using
\[
E=\frac{\partial\mathscr l}{\partial u}\cdot u-\mathscr l,
\qquad
p=\mathscr l-\rho\frac{\partial\mathscr l}{\partial\rho}
-s\frac{\partial\mathscr l}{\partial s},
\]
we get,
$$
\begin{aligned}
&v_s\left[\!\left[E\right]\!\right]-\left[\!\left[(E+p) u+T j_s\right]\!\right]\cdot n\\
&= v_s\left[\!\left[\frac{\partial \mathscr l}{\partial  u}\cdot u-\mathscr l\right]\!\right]-\left[\!\left[\left(\frac{\partial \mathscr l}{\partial  u}\cdot u- \rho\frac{\partial \mathscr l}{\partial \rho}-s\frac{\partial \mathscr l}{\partial s}\right) u+T j_s\right]\!\right]\cdot n.
\end{aligned}
$$
By substituting the jump relations obtained in Theorem \ref{RHvar_thermo}, we get
$$
\begin{aligned}
v_s\left[\!\left[\frac{\partial \mathscr l}{\partial  u}\cdot u\right]\!\right]=&\left[\!\left[\left(\frac{\partial \mathscr l}{\partial  u}\cdot u-\varsigma\frac{\partial \mathscr l}{\partial s}\right) u\right]\!\right]\cdot n-\left[\!\left[\rho  u\cdot\nabla w\right]\!\right]v_s+\left[\!\left[\rho  u\cdot\nabla w  u\right]\!\right]\cdot  n\\
&+\left[\!\left[ u \cdot \nabla\gamma\; j_s\right]\!\right]\cdot n-\left[\!\left[(s-\varsigma) u \cdot \nabla\gamma\right]\!\right]v_s+\left[\!\left[(s-\varsigma) u \cdot \nabla\gamma\; u\right]\!\right]\cdot  n.
\end{aligned}
$$
Using the continuity across $\Gamma$ of
$\partial_t w,\, \nabla w,\, 
\partial_t\gamma,\,\nabla\gamma,$
all terms reorganize into fluxes of the form
\[
\left[\!\left[\rho D_t w\,u\right]\!\right]\cdot n,
\qquad
\left[\!\left[(s-\varsigma)D_t\gamma\,u\right]\!\right]\cdot n,
\qquad
\left[\!\left[D_t\gamma\, j_s\right]\!\right]\cdot n,
\]
which cancel identically by the Rankine–Hugoniot conditions derived from the variational principle:
$$
0=v_s\left[\!\left[\rho D_t w\right]\!\right]+\left[\!\left[\rho\frac{\partial \mathscr l}{\partial \rho} u\right]\!\right]\cdot n-\left[\!\left[\rho  u\cdot\nabla w\right]\!\right]v_s+\left[\!\left[\rho  u\cdot\nabla w  u\right]\!\right]\cdot  n
$$
and similarly, 
$$
v_s\left[\!\left[(s-\varsigma) D_t \gamma\right]\!\right]+\left[\!\left[(s-\varsigma)\frac{\partial \mathscr l}{\partial s} u\right]\!\right]\cdot n+\left[\!\left[ u \cdot \nabla\gamma\; j_s\right]\!\right]\cdot n-\left[\!\left[(s-\varsigma) u \cdot \nabla\gamma\right]\!\right]v_s+\left[\!\left[(s-\varsigma) u \cdot \nabla\gamma\; u\right]\!\right]\cdot  n=\left[\!\left[D_t\gamma\; j_s\right]\!\right]\cdot n.
$$

All contributions therefore cancel, and hence $v_s\left[\!\left[E\right]\!\right]-\left[\!\left[(E+p) u+T j_s\right]\!\right]\cdot n=0$ as claimed.
\end{proof}

From this lemma, the following result follows immediately.

\begin{theorem}
The Euler–Lagrange equations associated with the variational principle \eqref{genLagrangian_thermo}--\eqref{Constraints_js} imply conservation of energy, i.e.,
\begin{equation}\label{time_change_E_lambda_thermo}
\frac{d}{dt}\mathscr E=0.
\end{equation}
\end{theorem}

\begin{proof}
Note that the boundary contribution in \eqref{time_change_E_thermo} is
\begin{equation}\label{E_cons_js}
-\int_{\partial M} p u\cdot n \,{\rm d}a
=0.
\end{equation}
Using Proposition~\ref{prop:energy_thermo}, the time derivative of the total energy reduces to an integral over the shock surface $\Gamma_t$,
\[
\frac{d}{dt}\mathscr E=\int_{\Gamma_t} v_s\left[\!\left[E\right]\!\right]-\left[\!\left[(E+p) u+T j_s\right]\!\right]\cdot n\,{\rm d}a.
\]
Hence by Lemma~\ref{lem:RH_E_lambda_thermo}, we get \eqref{E_cons_js}.
\end{proof}

\section{Conclusion}

In this paper, we have developed an extension of Hamilton’s principle
\begin{equation}\label{HP_concl}
\delta \int_{t_0}^{t_1} L(q,\dot q)\, {\rm d}t = 0,
\end{equation}
which is classically known to yield smooth solutions of compressible fluid models, to the setting of flows exhibiting shock discontinuities. The proposed variational framework produces, in a unified manner, the governing equations in the bulk, the associated boundary conditions, and the Rankine--Hugoniot jump conditions across discontinuities, for both fixed- and free-boundary configurations.

We have considered two distinct physical settings. First, for barotropic fluids, where the absence of an entropy variable leads to energy dissipation at shocks, we incorporate this loss through a dissipative mechanism. Second, for the full compressible Euler system, shock effects are instead captured through entropy production, reflecting the conversion of dissipated energy into internal energy. Our framework also naturally accommodates irreversible processes in the bulk, such as heat conduction.

These two settings are treated via different extensions of Hamilton’s principle:
\begin{itemize}
    \item[(1)] In the barotropic case, the action functional is augmented by a dissipative potential that encodes the energy loss at the shock. The contribution of this potential is localized on the shock surface. Formally, \eqref{HP_concl} is replaced by
    \[
    \delta \int_{t_0}^{t_1} \bigl( L(q,\dot q) + \mathcal{V}(q) \bigr)\, {\rm d}t = 0,
    \]
    for a suitable choice of dissipative potential $\mathcal{V}$.

    \item[(2)] In the full compressible case, we employ a variational formulation of nonequilibrium thermodynamics, which captures the mechanism by which energy dissipated at shocks is converted into internal energy, leading to entropy production. This results in a modified variational principle of the form
    \[
    \delta \int_0^T \Bigl( L(q,\dot q, S) + (S - \Sigma)\,\dot{\Gamma} \Bigr)\, {\rm d}t = 0,
    \]
    subject to the phenomenological and variational constraints
    \begin{equation}%\label{constraints}
    \frac{\partial L}{\partial S}\,\delta \Sigma
    =
    \langle F, \delta q\rangle,
    \qquad
    \frac{\partial L}{\partial S}\,\dot \Sigma
    =
    \langle F, \dot q\rangle,
    \end{equation}
    for a suitable dissipative force $F$.
\end{itemize}

The approach developed here provides a flexible and systematic framework applicable to a broad class of fluid models. Owing to its Lagrangian structure, it can be naturally extended to include additional advected variables and more complex physical effects. These features make it a promising foundation for further developments, both in modeling and in the design of structure-preserving numerical methods, which will be explored in future work.

\paragraph{Acknowledgements.} Fran\c{c}ois Gay-Balmaz was supported by a startup grant from Nanyang Technological University. Fran\c{c}ois Gay-Balmaz and Cheng Yang were supported by the National Research Foundation Singapore (NRF) core funding `Fusion Science for Clean Energy'.

\begin{appendices}\label{App}

\section{Background on Rankine–Hugoniot conditions and weak solutions of conservation Laws} \label{sec: lemma}

In the literature, the Rankine--Hugoniot relation are usually derived from the conservation laws in the one dimensional case. In \S\ref{A1} we prove these relations in higher dimension which is more suitable for our study. In \S\ref{A2}, we recall the definition of weak solutions, see \cite{evan}, which are helpful for the proofs of our main theorems.

\subsection{The Rankine--Hugoniot conditions and conservation laws}\label{A1}

The following lemma shows the relationship between conservation laws and Rankine--Hugoniot jump conditions across moving interfaces. Throughout this section, $M$ denotes a manifold, and $\Gamma \subset M$ is a codimension-one hypersurface separating $M$ into two open submanifolds $M_+$ and $M_-$. We write $[\![\cdot]\!]$ for the jump across $\Gamma$ and denote by $v_s$ the normal speed of $\Gamma$.
\begin{lemma}\label{lem:cons_law_RH}
Let $U$ be a vector-valued function and $F(U)$ be a matrix-valued function, both are smooth on $M_\pm$. Then the following statements are equivalent.

\medskip
\noindent
(i) For every fixed domain $D \subseteq M$,
\[
\frac{d}{dt}\int_D U \, {\rm d}x
=
-\int_{\partial D} F(U)\cdot n \, {\rm d}a .
\]

\medskip
\noindent
(ii) The pair $(U,F)$ satisfies
\[
\left\{
\begin{aligned}
&\partial_t U_\pm + \nabla \cdot F_\pm(U_\pm) = 0
\qquad \text{in } M_\pm,\\[0.3em]
&[\![F(U)]\!]\cdot n
=
[\![U]\!] \, v_s
\qquad \text{on } \Gamma .
\end{aligned}
\right.
\]
\end{lemma}

\begin{proof}
\noindent \textit{(i)} $\Rightarrow$ \textit{(ii)}.
Since $U_\pm$ and $F_\pm(U_\pm)$ are smooth in $M_\pm$, applying the divergence theorem on any fixed subdomain
$D_\pm \subseteq M_\pm$ at time $t$ yields
$$
\frac{d}{dt}\int_{D_\pm} U_\pm \, {\rm d}x
=
\int_{D_\pm} \partial_t U_\pm \, {\rm d}x,
\qquad
\int_{\partial D_\pm} F_\pm(U_\pm)\cdot n \, {\rm d}a
=
\int_{D_\pm} \nabla\cdot F_\pm(U_\pm) \, {\rm d}x.
$$
Thus $\partial_t U_\pm + \nabla\cdot F_\pm(U_\pm)=0$ in $M_\pm$.

Next, let $\mathfrak D$ be a small fixed domain intersecting $\Gamma$ along a smooth patch $\gamma$.
Decomposing $\mathfrak D=\mathfrak D_+\cup\mathfrak D_-$, we obtain
$$
\begin{aligned}
\frac{d}{dt}\int_{\mathfrak D} U\, {\rm d}x
&=\frac{d}{dt}\left(\int_{\mathfrak D_+} U_+\, {\rm d}x
+\int_{\mathfrak D_-} U_- \, dx\right)\\
&=\int_{\mathfrak D_+}\frac{\partial U_+}{\partial t}\, {\rm d}x
+\int_{\gamma} U_+\, u_s\cdot n_+\, {\rm d}a
+\int_{\mathfrak D_-}\frac{\partial U_-}{\partial t}\, {\rm d}x
+\int_{\gamma} U_-\, u_s\cdot n_-\, {\rm d}a\\
&=\int_{\mathfrak D_+}-\nabla\cdot F_+(U_+)\, {\rm d}x
+\int_{\mathfrak D_-}-\nabla\cdot F_-(U_-)\, {\rm d}x
+\int_{\gamma}[\![U]\!]\, v_s\, {\rm d}a\\
&=-\int_{\partial \mathfrak D_+}F_+(U_+)\cdot n_+\, {\rm d}a
-\int_{\partial \mathfrak D_-}F_-(U_-)\cdot n_-\, {\rm d}a
+\int_{\gamma}[\![U]\!]\, v_s\, {\rm d}a\\
&=-\int_{\partial \mathfrak D}F(U)\cdot n\, {\rm d}a
-\int_{\gamma}[\![F(U)]\!]\cdot n\, {\rm d}a
+\int_{\gamma}[\![U]\!]\, v_s\, {\rm d}a.
\end{aligned}
$$
Since
\(
\frac{d}{dt}\int_{\mathfrak D}U
=-\int_{\partial \mathfrak D}F(U)\cdot n,
\)
the integrand on $\gamma$ must vanish, yielding the Rankine--Hugoniot condition
$$
[\![F(U)]\!]\cdot n
=
[\![U]\!]\, v_s,
\qquad \text{on } \Gamma,
$$
since $\gamma$ is arbitrary in $\Gamma$.

\medskip

\noindent \textit{(ii)} $\Rightarrow$ \textit{(i)}.
If $D\subseteq M_\pm$, the identity follows directly from the divergence theorem,
$$
\frac{d}{dt}\int_D U \, {\rm d}x
=\int_D \frac{\partial U}{\partial t}\,{\rm d}x
=-\int_D \nabla\cdot F(U)\,{\rm d}x
=-\int_{\partial D}F(U)\cdot n\, {\rm d}a.
$$

If $D$ intersects $\Gamma$, the same decomposition as above shows that the Rankine--Hugoniot condition
exactly cancels the interface contribution, yielding the global conservation law
$$
\begin{aligned}
\frac{d}{dt}\int_D U \, {\rm d}x
&=\frac{d}{dt}\left(\int_{D_+} U_+ \, {\rm d}x
+\int_{D_-} U_- \, {\rm d}x\right)\\
&=\int_{D_+}\frac{\partial U_+}{\partial t}\,{\rm d}x
+\int_{\gamma} U_+\, v_s\cdot n_+\, {\rm d}a
+\int_{D_-}\frac{\partial U_-}{\partial t}\,{\rm d}x
+\int_{\gamma} U_-\, v_s\cdot n_-\, {\rm d}a\\
&=\int_{D_+}-\nabla\cdot F_+(U_+)\,{\rm d}x
+\int_{D_-}-\nabla\cdot F_-(U_-)\,{\rm d}x
+\int_{\gamma}[\![U]\!]\, v_s\, {\rm d}a\\
&=-\int_{\partial D_+}F_+(U_+)\cdot n_+\, {\rm d}a
-\int_{\partial D_-}F_-(U_-)\cdot n_-\, {\rm d}a
+\int_{\gamma}[\![U]\!]\, v_s\, {\rm d}a\\
&=-\int_{\partial D}F(U)\cdot n\, {\rm d}a
-\int_{\gamma}[\![F(U)]\!]\cdot n{\rm d}a
+\int_{\gamma}[\![U]\!]\, v_s\, {\rm d}a.
\end{aligned}
$$
Since
\(
[\![F(U)]\!]\cdot n
=
[\![U]\!]\, v_s
\)
on $\Gamma$, we obtain the conservation law (i).
\end{proof}

We now specialize to scalar conservation laws and moving domains.
\begin{lemma}\label{lem:cons_law_RH_2}
Suppose $\rho$ and $u$ satisfy
\[
\left\{
\begin{aligned}
&\partial_t \rho + \nabla\cdot(\rho u)=0
\qquad \text{in } M_\pm,\\
&[\![\rho u]\!]\cdot n
=
[\![\rho]\!]\, v_s
\qquad \text{on } \Gamma.
\end{aligned}
\right.
\]
Then the total mass in the time-dependent domain $M(t)$ is conserved:
\[
\frac{d}{dt}\int_{M(t)} \rho \, {\rm d}x = 0,
\]
provided that the boundary $\partial M(t)$ moves with the particle velocity.
\end{lemma}

\begin{proof}
We have
$$
\begin{aligned}
\frac{d}{dt}\int_{M(t)}\rho\,{\rm d}x
&=\frac{d}{dt}\left(\int_{M_+}\rho_+\,{\rm d}x+\int_{M_-}\rho_-\,{\rm d}x\right)\\
&=\int_{M_+}\frac{\partial \rho_+}{\partial t}\,dx
+\int_{\Gamma}\rho_+ v_s\cdot n_+\,{\rm d}a
+\int_{M_-}\frac{\partial \rho_-}{\partial t}\,{\rm d}x
+\int_{\Gamma}\rho_- v_s\cdot n_-\,{\rm d}a
+\int_{\partial M}\rho\, u\cdot n\,{\rm d}a\\
&=\int_{M_+}-\nabla\cdot(\rho_+ u_+)\,{\rm d}x
+\int_{M_-}-\nabla\cdot(\rho_- u_-)\,{\rm d}x
+\int_{\Gamma}[\![\rho]\!]v_s\,{\rm d}a
+\int_{\partial M}\rho\, u\cdot n\,{\rm d}a\\
&=-\int_{\partial M_+}\rho_+ u_+\cdot n_+\,{\rm d}a
-\int_{\partial M_-}\rho_- u_-\cdot n_-\,{\rm d}a
+\int_{\Gamma}[\![\rho]\!]v_s\,{\rm d}a
+\int_{\partial M}\rho\, u\cdot n\,{\rm d}a\\
&=-\int_{\partial M}\rho\, u\cdot n\,{\rm d}a
-\int_{\Gamma}[\![\rho u]\!]\cdot n\,{\rm d}a
+\int_{\Gamma}[\![\rho]\!]v_s\,{\rm d}a
+\int_{\partial M}\rho\, u\cdot n\,{\rm d}a.
\end{aligned}
$$
Since
\(
[\![\rho u]\!]\cdot n=[\![\rho]\!]v_s
\)
on $\Gamma$, we obtain
\(
\frac{d}{dt}\int_{M(t)}\rho\,{\rm d}x=0.
\)
\end{proof}

Conversely, conservation of mass for arbitrary material subdomains implies both the continuity equation and the Rankine--Hugoniot condition.
\begin{lemma}\label{lem:cons_law_RH_3}
Let $\rho$ be a scalar function such that
\[
\frac{d}{dt}\int_{D(t)} \rho \, {\rm d}x = 0
\]
for every material domain $D(t)\subseteq M$ whose boundary moves with the particle velocity. Then
\[
\left\{
\begin{aligned}
&\partial_t \rho + \nabla\cdot(\rho u)=0
\qquad \text{in } M_\pm,\\
&[\![\rho u]\!]\cdot n
=
[\![\rho]\!]\, v_s
\qquad \text{on } \Gamma.
\end{aligned}
\right.
\]
\end{lemma}

\begin{proof}
Taking $D(t)\subseteq M_\pm$ yields the continuity equation in each bulk region,
$$
\frac{d}{dt}\int_{D_\pm}\rho_\pm \, dx
=\int_{D_\pm}\frac{\partial \rho_\pm}{\partial t}\,{\rm d}x
+\int_{\partial D_\pm}\rho_\pm u_\pm\cdot n_\pm\,{\rm d}a
=0,
$$
and
$$
\int_{\partial D_\pm}\rho_\pm u_\pm\cdot n_\pm\,{\rm d}a
=\int_{D_\pm}\nabla\cdot(\rho_\pm u_\pm)\,{\rm d}x.
$$
Hence,
$$
\partial_t \rho + \nabla\cdot(\rho u)=0,
\qquad \text{in } M_\pm.
$$
Taking $D(t)$ intersecting $\Gamma$ and repeating the transport calculation shows that conservation of mass forces the Rankine--Hugoniot jump condition,
$$
\begin{aligned}
\frac{d}{dt}\int_{D}\rho\,{\rm d}x
&=\frac{d}{dt}\left(\int_{D_+}\rho_+\,{\rm d}x+\int_{D_-}\rho_-\,{\rm d}x\right)\\
&=\int_{D_+}\frac{\partial \rho_+}{\partial t}\,{\rm d}x
+\int_{\gamma}\rho_+ v_s\cdot n_+\,{\rm d}a
+\int_{D_-}\frac{\partial \rho_-}{\partial t}\,{\rm d}x
+\int_{\gamma}\rho_- v_s\cdot n_-\,{\rm d}a
+\int_{\partial D}\rho\, u\cdot n\,{\rm d}a\\
&=\int_{D_+}-\nabla\cdot(\rho_+ u_+)\,{\rm d}x
+\int_{D_-}-\nabla\cdot(\rho_- u_-)\,{\rm d}x
+\int_{\gamma}[\![\rho]\!]v_s\,{\rm d}a
+\int_{\partial D}\rho\, u\cdot n\,{\rm d}a\\
&=-\int_{\partial D_+}\rho_+ u_+\cdot n_+\,{\rm d}a
-\int_{\partial D_-}\rho_- u_-\cdot n_-\,{\rm d}a
+\int_{\gamma}[\![\rho]\!]v_s\,{\rm d}a
+\int_{\partial D}\rho\, u\cdot n\,{\rm d}a\\
&=-\int_{\partial D}\rho\, u\cdot n\,{\rm d}a
-\int_{\gamma}[\![\rho u]\!]\cdot n
+\int_{\gamma}[\![\rho]\!]v_s\,{\rm d}a
+\int_{\partial D}\rho\, u\cdot n\,{\rm d}a.
\end{aligned}
$$
Since $\frac{d}{dt}\int_{D(t)}\rho\,{\rm d}x=0$ and $\gamma$ is an arbitrary small part of the shock surface $\Gamma$, we conclude that
$$
[\![\rho u]\!]\cdot n
=
[\![\rho]\!]v_s,
\qquad \text{on } \Gamma.
$$
\end{proof}

\begin{remark}
The Rankine--Hugoniot conditions above are obtained directly from the conservation
form of the Euler equations and express the balance of fluxes across a moving
discontinuity.
This derivation should be contrasted with that of Proposition~\ref{weaksolRH} below,
where analogous jump conditions are recovered from the definition of weak solutions
by testing against smooth compactly supported functions.
\end{remark}

\subsection{The Rankine--Hugoniot conditions and definition of weak solutions}\label{A2}

We next show how the Rankine--Hugoniot jump condition can be derived directly from
the definition of weak solutions. To illustrate the argument in a transparent
setting, we consider the scalar conservation law
\begin{equation}\label{diveq_model}
\partial_t \rho + \nabla \cdot K(\rho) = 0,
\end{equation}
where $\rho=\rho(t,x)$ is a scalar function defined on spacetime
$(t,x)\in \mathbb{R}\times \mathbb{R}^n$, and
$K(\rho)=K(\rho(t,x)) \in \mathbb{R}^n$ is a given flux function.

Let $\Omega\subset\mathbb{R}^{n+1}$ be a spacetime domain and let
$\Sigma\subset\Omega$ be a smooth codimension-one hypersurface representing a shock.
We assume that $\Sigma$ separates $\Omega$ locally into two open subdomains
$\Omega_-$ and $\Omega_+$ and is described as the zero level set of a smooth function
$f$, namely
\[
\Sigma=\{(t,x)\in\Omega \mid f(t,x)=0\},
\]
with $\nabla_x f\neq 0$ on $\Sigma$.

\begin{lemma}\label{normspeed}
The spacetime hypersurface $\Sigma$ can be represented as a family of moving spatial
interfaces $\Gamma(t)$ propagating with normal velocity $v_s$.
If $\Gamma(t)$ is locally given by $f(t,x)=0$, then the normal propagation speed is
\begin{equation}
v_s=-\frac{\partial_t f}{|\nabla_x f|}.
\end{equation}
\end{lemma}

\begin{proof}
By $f(t,x)=0$, we get
$$
\partial_t f\, dt+\nabla_x f\, dx=0.
$$
Dividing by $|\nabla_x f|\,{\rm d}t$ on both sides gives
$$
\frac{\partial_t f}{|\nabla_x f|}
+\frac{\nabla_x f}{|\nabla_x f|}\cdot\frac{dx}{dt}=0.
$$
On the other hand, the normal speed satisfies
$v_s=\frac{dx}{dt}\cdot n
=\frac{\nabla_x f}{|\nabla_x f|}\cdot\frac{dx}{dt}$,
and the claim follows.
\end{proof}

We now derive the Rankine--Hugoniot condition for \eqref{diveq_model} using the weak
formulation.

\begin{proposition}\label{weaksolRH}
Let $\rho$ be a weak solution of \eqref{diveq_model} admitting a jump discontinuity
across $\Sigma$. Then $\rho$ satisfies the Rankine--Hugoniot condition
\begin{equation}\label{diveq_modelRH}
[\![\rho]\!]\, v_s = [\![K(\rho)]\!] \cdot n,
\end{equation}
where $[\![\cdot]\!]$ denotes the jump across $\Sigma$,
$v_s=-\frac{\partial_t f}{|\nabla_x f|}$ is the normal velocity, and
$n=\frac{\nabla_x f}{|\nabla_x f|}$ is the unit normal to $\Gamma(t)$.
\end{proposition}

\begin{proof}
By definition, $\rho$ is a weak solution of \eqref{diveq_model} if
\[
\int_{\mathbb{R}}\int_{\mathbb{R}^n}
\bigl(\rho\,\partial_t h+K(\rho)\cdot\nabla_x h\bigr)\,{\rm d}x\,{\rm d}t=0,
\qquad
\forall\,h\in C_c^\infty(\mathbb{R}\times\mathbb{R}^n).
\]
Decomposing the spacetime integral according to the partition
$\Omega=\Omega_-\cup\Sigma\cup\Omega_+$ yields
\begin{equation}\label{weaksol1}
\iint_{\Omega_-}(\rho\partial_t h+K(\rho)\cdot\nabla_x h)\,{\rm d}x\,{\rm d}t
+
\iint_{\Omega_+}(\rho\partial_t h+K(\rho)\cdot\nabla_x h)\,{\rm d}x\, {\rm d}t=0.
\end{equation}
Integrating by parts on each subdomain gives
\begin{equation}\label{weaksol2}
\begin{aligned}
\iint_{\Omega_\pm}(\rho\partial_t h+K(\rho)\cdot\nabla_x h)\,{\rm d}x\,{\rm d}t
&=
-\iint_{\Omega_\pm}
\left(\partial_t\rho+\nabla\cdot K(\rho)\right)h\,{\rm d}x\,{\rm d}t \\
&\quad
+\int_{\Sigma}
\left(\rho h\, n_t^{\pm}
+h\,K(\rho)\cdot n_x^{\pm}\right)\,{\rm d}a,
\end{aligned}
\end{equation}
where $N^\pm=(n_t^\pm,n_x^\pm)$ denotes the outward unit normal to $\Sigma$
relative to $\Omega_\pm$.

Since $\rho$ is smooth away from $\Sigma$, it satisfies
\begin{equation}\label{weaksol3}
\partial_t\rho+\nabla\cdot K(\rho)=0
\quad\text{in }\Omega_\pm,
\end{equation}
and the volume integrals vanish. Combining
\eqref{weaksol1}--\eqref{weaksol3} yields
\[
\int_{\Sigma}
\bigl(\rho^+ n_t+K(\rho^+)\cdot n_x\bigr)h\,{\rm d}a
-
\int_{\Sigma}
\bigl(\rho^- n_t+K(\rho^-)\cdot n_x\bigr)h\,{\rm d}a
=0,
\]
for all test functions $h$. Hence,
\begin{equation}\label{weaksol4}
\rho^+ n_t+K(\rho^+)\cdot n_x
=
\rho^- n_t+K(\rho^-)\cdot n_x,
\end{equation}
where $N=(n_t,n_x)=\frac{(\partial_t f,\nabla_x f)}{|(\partial_t f,\nabla_x f)|}$.

Rewriting \eqref{weaksol4} gives
\[
[\![\rho]\!]\,\partial_t f
=
-[\![K(\rho)]\!]\cdot\nabla_x f.
\]
Dividing by $|\nabla_x f|$ and using
$v_s=-\frac{\partial_t f}{|\nabla_x f|}$ and
$n=\frac{\nabla_x f}{|\nabla_x f|}$ yields \eqref{diveq_modelRH}.
\end{proof}
\end{appendices}

\end{document}